\documentclass{article}

\usepackage[hyphens]{url}
\usepackage{hyperref}
\hypersetup{breaklinks=true}
\usepackage{graphicx}          
\usepackage{amsmath} 
\usepackage{amssymb}  
\usepackage{amsthm}
\usepackage{bm}         
\usepackage{mathtools}          

\usepackage{algorithm}  
\usepackage{algpseudocode}
\usepackage[english]{babel}     

\usepackage{xcolor}             
\usepackage{psfrag}             
\usepackage{subcaption}
\usepackage{enumerate}
\usepackage{cite}               

\makeatletter
\renewcommand*\env@matrix[1][*\c@MaxMatrixCols c]{%
  \hskip -\arraycolsep
  \let\@ifnextchar\new@ifnextchar
  \array{#1}}
\makeatother

\newcommand{\figref}[1]{Fig.~\ref{#1}}

\newcommand{\secref}[1]{Section~\ref{#1}}
\newcommand{\remref}[1]{\textit{Remark~\ref{#1}}}
\newcommand{\assref}[1]{\textit{Assumption~\ref{#1}}}
\newcommand{\defref}[1]{\textit{Definition~\ref{#1}}}

\newcommand{\thmref}[1]{\textit{Theorem~\ref{#1}}}
\newcommand{\exmref}[1]{\textit{Example~\ref{#1}}}
\newcommand{\algoref}[1]{\textit{Algorithm~\ref{#1}}}

\newcommand{\Rset}{\mathbb{R}}
\newcommand{\Nset}{\mathbb{N}}
\newcommand{\Zset}{\mathbb{Z}}

\newcommand{\Bs}{\mathcal{B}}

\newcommand{\Es}{\mathcal{E}}

\newcommand{\Ks}{\mathcal{K}}

\newcommand{\Ns}{\mathcal{N}}
\newcommand{\Ps}{\mathcal{P}}
\newcommand{\Qs}{\mathcal{Q}}
\newcommand{\Rs}{\mathcal{R}}
\newcommand{\Ss}{\mathcal{S}}
\newcommand{\Us}{\mathcal{U}}
\newcommand{\Ws}{\mathcal{W}}
\newcommand{\Xs}{\mathcal{X}}

\newcommand{\diag}{\operatorname{diag}}
\newcommand{\post}{\operatorname{post}}
\newcommand{\Post}{\operatorname{Post}}

\newtheorem{defn}{Definition}
\newtheorem{exmp}{Example}
\newtheorem{assum}{Assumption}
\newtheorem{rem}{Remark}
\newtheorem{thm}{Theorem}

\newif\ifthesis
\thesistrue
\thesisfalse

\begin{document}

\begin{center}
  {\bf\Large Switched-Actuator Systems with Setup Times: Efficient Modeling, MPC, and Application to Hyperthermia Therapy\footnote{This research has been made possible by the Dutch Cancer Society and the Netherlands Organisation for Scientific Research (NWO) as part of their joint Partnership Programme: ``Technology for Oncology.'' This project is partially financed by the PPP Allowance made available by Top Sector Life Sciences \& Health. This work is also partially funded by the European Union via the IPaCT Project and by the German Federal Ministry of Education and Research ``MR-HIFU-Pancreas.''}}
  \\[5mm]
  {\large D.A. Deenen$^*$, E. Maljaars$^*$, L.C. Sebeke$^{**}$, B. de Jager$^*$, E.~Heijman$^{***}$, H. Gr\"{u}ll$^{**}$, W.P.M.H. Heemels$^*$}\\[3mm]
  $^*$ Control Systems Technology, Department of Mechanical Engineering, Eindhoven University of Technology, The Netherlands\\
  $^{**}$ University Hospital of Cologne, University of Cologne, Germany\\
  $^{***}$ Philips Research, The Netherlands
\end{center}

\section*{Abstract}
Switched-actuator systems with setup times (SAcSSs) are systems in which the actuator configuration has to be switched during operation, and where the switching induces non-negligible actuator downtime.
Optimally controlling SAcSSs requires the online solving of both a discrete actuator allocation problem, in which the switch-induced actuator downtime is taken into account, as well as an optimization problem for the (typically continuous) control inputs. Mixed-integer model predictive control (MI-MPC) offers a powerful framework for tackling such problems.
However, the efficient modeling of SAcSSs for MI-MPC is not straightforward, and real-time feasibility is often a major hurdle in practice.
It is the objective of this paper to provide an intuitive and systematic modeling procedure tailored to SAcSSs, which is specifically designed to allow for user-friendly controller synthesis, and to yield efficient MI-MPCs.
We apply these new results in a case study of large-volume magnetic-resonance-guided high-intensity focused ultrasound hyperthermia, which involves the heating of tumors (using real-valued local heating controls, as well as discrete range-extending actuator relocation during which no heating is allowed) to enhance the efficacy of radio- and chemotherapy.


\section{Introduction}\label{mip:sec:intro}

A major motivation for the system-theoretic contributions of this paper is large-volume magnetic-resonance-guided high-intensity focused ultrasound (MR-HIFU) hyperthermia \cite{Maloney2015}.
In hyperthermia treatments, malignant tissue is heated to about 42~$^\circ$C for 60 to 90 minutes. Using MR-HIFU, the heating is applied locally and noninvasively via ultrasound waves, based on real-time temperature measurements obtained with an MRI scanner.
The heating sensitizes the tissue to the effects of chemo- and radiotherapy without adding any undesirable (toxic) side effects \cite{Oei2015,Mallory2016}. Thereby, MR-HIFU hyperthermia enables significantly higher success rates and allows for a considerable reduction of the unwanted side effects of chemo- and radiotherapeutic cancer treatments \cite{Datta2015,Issels2018}.

MR-HIFU allows for powerful heating with millimeter-accurate steering. Using MPC, HIFU-mediated thermal therapies of optimal quality can be realized, while respecting actuator and safety constraints, see
\ifthesis
\else
\cite{Arora2007,DeBever2014,Hensley2015,Sebeke2019,Deenen2020a,Deenen2020}.
\fi
Unfortunately, this high accuracy comes with limited spatial heating range. Thereby, only small tumors ($\leq$16~mm diameter) can be treated using a stationary actuator. For larger tumors, the actuator itself must be relocated, during which no heating can occur \cite{Tillander2016}. As the set of admissible positions must be discrete and finite, the resulting system can be described as a switched system, where the input model differs depending on the actuator location. However, a first distinctive property of this hyperthermia system, compared to typical switched systems, is that the time required for actuator relocation, and thus a mode switch, is non-negligible. Second, any mode switch induces nonzero actuator downtime, during which the system itself keeps evolving in time according to its unforced dynamics. Clearly, for the optimal control of the heating process in large-volume MR-HIFU hyperthermia, these two features have to be incorporated in the controller design, determining online the continuous local heating as well as the discrete actuator positioning.

The essential features recognized in large-volume MR-HIFU hyperthermia, i.e., dealing with dynamical systems in which the (finite number of) actuator configurations have to be switched and where the switching takes significant time during which no active control is possible, are of a general nature as they can be identified in many other applications. Indeed, one could think about manufacturing systems in which machine reconfiguration takes time, see \cite{Allahverdi1999,Allahverdi2008,Allahverdi2015} for comprehensive surveys, but also
agents in agriculture (e.g., drones, fertilizers, irrigation systems)
which must serve multiple
(sub)fields resulting in significant field-to-field travel times \cite{Cobbenhagen2018,Schoonen2019}, or the coordinated deployment of fire-fighting units for wildfire management \cite{Donovan2003,Haight2007,Petrovic2012}.
Motivated by this range of applications,
we formally introduce
\emph{switched-actuator systems with setup times} (SAcSSs), exhibiting the mentioned features, as the first contribution of this work.

As a second contribution, we address the design of easy-to-derive and efficient MPC schemes for SAcSSs. In particular, we desire a natural and systematic modeling procedure for SAcSSs, yielding a compact model, i.e., with a small number of integer variables, that can be directly integrated into a mixed-integer MPC (MI-MPC) \cite{Borrelli2017} setup. The small number of integers is motivated by the desire to keep the computational burden limited, thereby facilitating real-time implementation of the resulting MPC schemes.
The mixed-integer programming (MIP) compatibility could be achieved by describing a SAcSS as a mixed logical dynamical (MLD) system \cite{Bemporad1999}. However, efficiently capturing the actuator switching behavior including setup times is not straightforward. Clearly, one could attempt to describe a SAcSS as a discrete hybrid automaton (DHA) \cite{Torrisi2004}, which using the modeling language HYSDEL could then be automatically converted into MLD form amenable for online optimization. Unfortunately, due to the generic nature of the DHA framework, it may be unclear how to best incorporate SAcSSs' actuator switching, thereby potentially leading to models containing too many Boolean variables, which would have negative consequences for the computational complexity of the online MPC problem. For example, a SAcSS could first be cast as a constrained switched (linear) system \cite{Liberzon2003,Philippe2016} before deriving its equivalent DHA and (using HYSDEL) MLD forms, yielding a similar result as the lifting approach in \cite{Subramanian2012}.
However, incorporating an actuator switch's setup time would in these cases require the inclusion of as many additional ``transitioning modes'' (i.e., duplicates of the zero-input mode corresponding to the unforced dynamics) as the switch's setup time. In turn, this results in a model with many integer (Boolean) variables, leading indeed to high-complexity representations. Therefore, as a second contribution, we propose an intuitive and convenient modeling procedure tailored to SAcSSs, specifically designed to yield compact models that are directly suitable for efficient MI-MPC setups.

The final contribution of this work consists of a large-volume MR-HIFU hyperthermia cancer therapy case study. Herein, we apply the novel methods to obtain a SAcSS model and MI-MPC, of which the performance and computational efficiency are validated.

A preliminary version of this work was published in \cite{Deenen2020a}, which discusses only a simulation study of MI-MPC for large-volume MR-HIFU hyperthermia (different from the one included here) as an illustrative proof-of-concept, but does not provide the general modeling procedure for SAcSSs
presented in this paper. In fact, the class of SAcSSs is formally defined here for the first time. Moreover, in this work we formalize the key concepts of the proposed modeling framework, provide rigorous proofs, discuss the general MI-MPC setup for SAcSSs, and investigate its improved computational efficiency, all of which was not included in \cite{Deenen2020a}.

The remainder of the paper is organized as follows. In \secref{mip:sec:sacss} we formally define SAcSSs, for which in \secref{mip:sec:modeling} we present the intuitive and compact modeling procedure. The resulting MI-MPC will be given in \secref{mip:sec:mipmpcsacss}. In \secref{mip:sec:mrhifu}, the large-volume MR-HIFU hyperthermia case study is presented.
Finally, \secref{mip:sec:conclusion} summarizes the key contributions of this work.

\textbf{Notation.} The real, integer, and natural numbers (including zero) are denoted by $\Rset$, $\Zset$, and $\Nset$, respectively. Given a set $\Ss\subseteq\Rset$ and values $a,b\in\Ss$, we use $\Ss_{>a}$, $\Ss_{\geq a}$, $\Ss_{<a}$, and $\Ss_{\leq a}$ to denote the subset of $\Ss$ of which the elements satisfy the condition in the subscript, and we define $\Ss_{[a,b]}=\{s\in\Ss\mid a\leq s\leq b\}$ (hence $\Ss_{[a,b]}=\emptyset$ if $a>b$). $\Ps(\Ss)$ denotes the power set of $\Ss$, being the collection of all subsets of $\Ss$, and $|\Ss|$ denotes the cardinality of $\Ss$. For $n\in\Nset_{>0}$, we denote the $n$-dimensional identity matrix by $I_n$, all-ones vector by $1_n$, and all-zeros vector by $0_n$ (subscript may be omitted if the dimension is clear from context)%
\ifthesis
\else
, and the block diagonal matrix with blocks $A_1,\ldots,A_n$ by $\diag(A_1,\dots,A_n)$.
\fi

\ifthesis
\else
\section{Switched-Actuator System with Setup Times}\label{mip:sec:sacss}
\fi

In this section, we formally introduce the class of SAcSSs using a state-space system, a weighted graph, and an actuator selector function.

\begin{defn}\label{mip:def:edgeweightedautomaton}\cite{Gross2013}
  A simple (i.e., without self-loops or multi-arcs) arc-weighted directed graph (digraph) is defined by the triple $\Gamma=(\Qs,\Es,s)$, where $\Qs=\Nset_{[1,N_q]}$ is the set of $N_q\in\Nset_{>0}$ nodes, $\Es\subseteq\Qs^2\setminus\{(q,q)\mid q\in\Qs\}$ is the set of arcs (directed edges), and $s:\Es\to\Nset$ is an arc-weighting function. If $\Es=\Qs^2\setminus\{(q,q)\mid q\in\Qs\}$, we say that the simple digraph is complete. If $s(q,\tilde{q})\leq s(q,\bar{q})+s(\bar{q},\tilde{q})$ for all distinct nodes $q,\tilde{q},\bar{q}\in\Qs$ of a complete simple digraph, we say that the digraph satisfies the triangle inequality.
\end{defn}

\begin{exmp}\label{mip:exm:sacssautomaton}
  A complete simple arc-weighted digraph $\Gamma$ with $N_q=4$ nodes and weighting function
  \begin{equation}\label{mip:eq:exms}
    s(q,\tilde{q}) = \left\{\begin{array}{ll} 1, &\text{if } (q,\tilde{q})\in\{(1,3),(3,1)\}, \\
    3, &\text{if } (q,\tilde{q})\in\{(2,4),(4,2)\}, \\
    2, &\text{otherwise}, \end{array}\right.
  \end{equation}
  satisfying the triangle inequality, is depicted in \figref{mip:fig:exmsacss}. Note that although here $s(q,\tilde{q})=s(\tilde{q},q)$ for all $(q,\tilde{q})\in\Es$, this is not necessary.
  \begin{figure}[t]
    \centering
    \includegraphics[width=0.35\textwidth,clip]{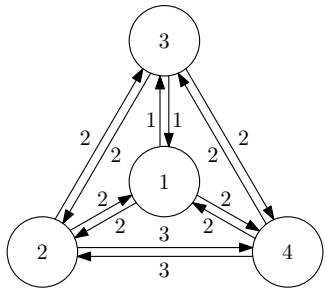}
    \caption{Example of a complete simple arc-weighted digraph. The numbered circles represent the nodes, and the numbered arrows represent the weighted arcs.}
    \label{mip:fig:exmsacss}
  \end{figure}%
\end{exmp}

\begin{defn}\label{mip:def:sacss}
  A switched-actuator system with setup times (SAcSS) is defined by the quintuple $\Sigma=(\Xs,\Us,f,\Gamma,\Phi)$, where $\Xs\subseteq\Rset^{n_x}$ with $n_x\in\Nset_{>0}$ is the state space, $\Us\subseteq\Rset^{n_u}$ with $n_u\in\Nset_{>0}$ is the input space, $f:\Xs\times\Us\to\Xs$ is a state transition function, $\Gamma$ is a complete simple arc-weighted digraph that satisfies the triangle inequality, and $\Phi:\Qs\cup\Es\to\Ps(\Nset_{[1,n_u]})$ is an actuator selector function for which it holds that
  \begin{equation}\label{mip:eq:Phi}
    \Phi(\sigma) \subseteq \left\{\begin{array}{ll} \Nset_{[1,n_u]}, & \text{if }\sigma\in\Qs, \\ \emptyset, & \text{if }\sigma\in\Es, \end{array}\right.
  \end{equation}
  where $\sigma\in\Qs\cup\Es$ indicates the actuator state.
\end{defn}

Some explanation regarding \defref{mip:def:sacss} is in order. In a SAcSS, the plant dynamics are given by
\begin{equation}\label{mip:eq:sacss}
  x_{k+1} = f(x_k,u_k)
\end{equation}
with state $x_k\in\Xs$ and input $u_k\in\Us$ at discrete time $k\in\Nset$. Note that a solution of \eqref{mip:eq:sacss} is completely defined by an initial condition $x_0\in\Xs$ and an input sequence $\bm{u}_{[0,K]}=(u_0,\dots,u_K)\in\Us^{K+1}$ with $K\in\Nset$. The input sequence, however, cannot be freely chosen, but must satisfy the restrictions resulting from the actuator switching, as will be defined in the remainder of this section.
First, consider the arc-weighted digraph $\Gamma$, which
describes the actuator switching behavior. That is, the nodes in $\Gamma$ represent the $N_q$ \emph{operational} actuator modes $q\in\Qs$, where in each mode a different subset of input channels is allowed to be nonzero. An arc $(q,\tilde{q})\in\Es$ represents the possible actuator mode switch from $q\in\Qs$ to $\tilde{q}\in\Qs$.
The arc weight $s(q,\tilde{q})\in\Nset$ represents the setup time associated with the mode switch $(q,\tilde{q})\in\Es$, where the setup time is the number of discrete time instants during which active control is disabled for \eqref{mip:eq:sacss}, i.e., during which all inputs must be zero. Hence, the weighting function $s$ describes the setup times. We denote the actuator state at time $k$ by $\sigma_k\in\Qs\cup\Es$, for which it holds that $\sigma_k=q$ when the actuator mode $q\in\Qs$ is active, and $\sigma_k=(q,\tilde{q})\in\Es$ when the SAcSS is in the process of switching from mode $q\in\Qs$ to $\tilde{q}\in\Qs$. Then, the actuator selector function $\Phi$ in \eqref{mip:eq:Phi} selects the subset of input channels that are allowed to be nonzero on the basis of the actuator state. In particular, when at time $k$ the actuator is operational in mode $q\in\Qs$, i.e., $\sigma_k=q$, the input channels $u_{k,i}$ with $i\in\Phi(\sigma_k)=\Phi(q)\subseteq\Nset_{[1,n_u]}$ can be nonzero, whereas $u_{k,j}=0$ for all $j\in\Nset_{[1,n_u]}\setminus\Phi(q)$.
When, on the other hand, the SAcSS is in the process of switching from a mode $q\in\Qs$ to another mode $\tilde{q}\in\Qs$, and thus $\sigma_k=(q,\tilde{q})\in\Es$, all inputs are disabled, i.e., $u_{k,i}=0$ for all $i\in\Nset_{[1,n_u]}\setminus\Phi(\sigma_k)=\Nset_{[1,n_u]}$ in this case as $\Phi(\sigma_k)=\Phi(q,\tilde{q})=\emptyset$.

Next, we define what actuator sequences are considered to describe proper switching behavior. To this end, we first introduce the
\emph{destination} mode corresponding to $\sigma_k$
denoted
by $\post(\sigma_k)$, where $\post:\Qs\cup\Es\to\Qs$ is defined as $\post(\sigma)=q$ if $\sigma=q\in\Qs$ or $\sigma=(\tilde{q},q)\in\Es$ (note that in case of no switch, the destination mode is simply defined as the current mode).
Second,
if the actuator state describes a mode switch, i.e., $\sigma_k\in\Es$, we denote the set of possible actuator states upon finishing the switch by $\Post(\sigma_k)$, where $\Post:\Es\to\Qs\cup\Es$ is defined as $\Post(q,\tilde{q}) = \{\tilde{q}\} \cup \{(\tilde{q},\bar{q})\in\Es \mid s(\tilde{q},\bar{q})>0\} \cup \{\bar{q}\in\Qs \mid (\tilde{q},\bar{q})\in\Es \text{ and } s(\tilde{q},\bar{q})=0\}$. In other words, given $\sigma_k=(q,\tilde{q})\in\Es$, the SAcSS can either stay in the destination mode $\tilde{q}$ upon arrival at time $l\in\Nset_{>k}$, indicated by $\sigma_l=\tilde{q}$,
or immediately start another switch $(\tilde{q},\bar{q})\in\Es$, resulting in $\sigma_l=(\tilde{q},\bar{q})$ or $\sigma_{l}=\bar{q}$ in case of nonzero or zero setup time $s(\tilde{q},\bar{q})$, respectively.

\begin{exmp}\label{mip:exm:postPost}
  Consider a SAcSS with $\Gamma$ as in \exmref{mip:exm:sacssautomaton}. Then for $\sigma=(1,2)$, the destination mode is given by $\post(\sigma) = 2$, and the destination actuator state set is $\Post(\sigma) = \{2,(2,1),(2,3),(2,4)\}$.
\end{exmp}

Now, we can define what actuator sequences describe proper switching behavior for a SAcSS $\Sigma$, which we refer to as being \emph{$\Sigma$-admissible} (or \emph{admissible} in case $\Sigma$ is clear from context), as follows.

\begin{defn}\label{mip:def:admissiblesigma}
  An actuator sequence $\bm{\sigma}_{[0,K]}=(\sigma_0,\dots,\sigma_K)\in(\Qs\cup\Es)^{K+1}$, $K\in\Nset$, is called \emph{$\Sigma$-admissible}, if
  \begin{enumerate}[(A)]
    \item\label{mip:def:admissiblesigma1} $\sigma_0\in\Qs$, and,
    \item\label{mip:def:admissiblesigma2} if $\sigma_k\neq\sigma_{k-1}$ with $\sigma_{k-1}\in\Qs$ and $\sigma_k\in\Qs\cup\Es$ for some $k\in\Nset_{[1,K]}$, then $\sigma_l=\tilde{\sigma}=(\sigma_{k-1},\post(\sigma_k))$ for all $l\in\Nset_{[k,\min\{k+s(\tilde{\sigma})-1,K\}]}$, and, if $k+s(\tilde{\sigma})\leq K$, also $\sigma_{k+s(\tilde{\sigma})}\in\Post(\tilde{\sigma})$.
  \end{enumerate}
\end{defn}

In other words, \eqref{mip:def:admissiblesigma1} specifies that the SAcSS must initialize in an operational actuator mode $\sigma_0\in\Qs$. By \eqref{mip:def:admissiblesigma2}, any actuator mode switch that is started must also be completed while adhering to its setup time. In particular, if at some time $k\in\Nset_{[1,K]}$ the SAcSS starts the switch from $q\in\Qs$ to $\tilde{q}\in\Qs$, we have $\sigma_{k-1}=q$, and either $\sigma_k=\tilde{q}$ or $\sigma_k=(q,\tilde{q})$ in case of zero or nonzero setup time $s(q,\tilde{q})$, respectively, and hence $\sigma_k\neq\sigma_{k-1}$. Moreover, since $\sigma_k\in\{\tilde{q},(q,\tilde{q})\}$, we have $\post(\sigma_k)=\tilde{q}$, and hence the considered mode switch can be denoted by $\tilde{\sigma}=(\sigma_{k-1},\post(\sigma_k))$ with setup time $s(\tilde{\sigma})$. Thus, \eqref{mip:def:admissiblesigma2} indeed states that, once started, this switch $\tilde{\sigma}$ must be completed, by defining that the SAcSS must be in the process of switching $\sigma_l=\tilde{\sigma}$ for all times $l\in\Nset_{[k,k+s(\tilde{\sigma})-1]}$ if $k+s(\tilde{\sigma})\leq K$ (note that $\Nset_{[k,k+s(\tilde{\sigma})-1]}=\emptyset$ if $s(\tilde{\sigma})=0$, corresponding to an instantaneous mode switch), or for all $l\in\Nset_{[k,K]}$ if $k+s(\tilde{\sigma})>K$. In case of the former, upon reaching the destination mode $\tilde{q}$ at time $k+s(\tilde{\sigma})\leq K$, all possible options for $\sigma_{k+s(\tilde{\sigma})}$ are described by $\Post(\tilde{\sigma})$.

Recall that the control inputs of \eqref{mip:eq:sacss} must adhere to the restrictions resulting from the actuator switching via $\Phi(\sigma_k)$.
We define the $\Sigma$-admissibility of an actuator sequence in combination with an input sequence as follows.

\begin{defn}\label{mip:def:admissiblesigmau}
  Let a SAcSS $\Sigma$ be given. Then, a pair $(\bm{\sigma}_{[0,K]},\bm{u}_{[0,K]})$, $K\in\Nset$, consisting of an actuator sequence $\bm{\sigma}_{[0,K]}\in(\Qs\cup\Es)^{K+1}$ and a control input sequence $\bm{u}_{[0,K]}=(u_0,\dots,u_K)\in\Us^{K+1}$, is called \emph{$\Sigma$-admissible} if $\bm{\sigma}_{[0,K]}$ is a $\Sigma$-admissible actuator sequence, and for all $k\in\Nset_{[0,K]}$ it holds that $u_{k,i}=0$ for all $i\in\Nset_{[1,n_u]}\setminus\Phi(\sigma_k)$.
\end{defn}

To facilitate the automated model generation proposed in \secref{mip:sec:modeling}, we use the following standing
assumption.

\begin{assum}\label{mip:ass:modallydisjointinputs}
  For the actuator selector function $\Phi$, it holds that $\Phi(q)\cap\Phi(\tilde{q})=\emptyset$ for all $q,\tilde{q}\in\Qs$, $q\neq\tilde{q}$.
\end{assum}

\assref{mip:ass:modallydisjointinputs} states that the SAcSS has disjoint input channels per actuator mode. This is without loss of generality, as this can always be achieved by simple duplication of the shared input channels%
\ifthesis
\else
.
\fi
In the SAcSS model and the resulting MI-MPC proposed in this paper, this property allows for enabling and disabling each mode's inputs using their corresponding Boolean variables. Combined with the digraph $\Gamma$ being complete and satisfying the triangle inequality by \defref{mip:def:sacss}, this allows for systematically encoding $\Sigma$-admissibility into the mentioned Boolean variables using the setup times, and thereby the automated derivation of mixed-integer constraints in \secref{mip:subsec:MIconstraints} as part of our modeling setup. Note that the digraph being complete and satisfying the triangle inequality is not a restrictive property, as discussed next in \remref{mip:rem:completeandtriangle}.

\begin{rem}\label{mip:rem:completeandtriangle}
  By completeness of $\Gamma$, a switch from any actuator mode $q\in\Qs$ to any other mode $\tilde{q}\in\Qs$ is possible. This is without loss of generality when the SAcSS' true mode-switching behavior is described by a strongly connected digraph (i.e., where for all nodes $q,\tilde{q}\in\Qs$, $q\neq\tilde{q}$, there exists a directed path connecting $q$ to $\tilde{q}$ and vice versa \cite{Gross2013}). That is, in case $(q,\tilde{q})\notin\Es$ in the strongly connected digraph, we can simply define the additional arc $(q,\tilde{q})$ with weight $s(q,\tilde{q})\coloneqq\min_{N,\bm{q}_{[1,N]}\in\Qs^N}\sum_{i=1}^{N-1}s(q_i,q_{i+1})$, corresponding to the minimum-time mode sequence $\bm{q}_{[1,N]}=(q_1,\ldots,q_N)$ with $q_1=q$, $q_N=\tilde{q}$, and $(q_i,q_{i+1})\in\Es$ for all $i\in\Nset_{[1,N-1]}$. Next, satisfying the triangle inequality implies that any switch $(q,\tilde{q})\in\Es$ cannot induce a setup time larger than the cumulative setup time along any other sequence connecting $q$ to $\tilde{q}$, which without loss of generality represents the underlying assumption that each switch corresponds to changing the actuator configuration in the fastest way possible.
\end{rem}

In the upcoming two sections, we aim to achieve this paper's main objective of providing a convenient modeling and controller synthesis procedure for SAcSSs that leads to computationally efficient MI-MPCs.

\ifthesis
\else
\section{SAcSS Modeling Framework}\label{mip:sec:modeling}
\fi

In this section, we present an intuitive and systematic procedure for modeling SAcSSs, which directly yields a compact MIP-compatible model. This procedure is specifically designed to (a) avoid the need for introducing many auxiliary discrete (Boolean) variables, which would otherwise increase model complexity and thereby (likely) the corresponding MI-MPC's computational burden, and (b) allow for user-friendly model specification of a SAcSS $\Sigma$ as in \defref{mip:def:sacss}, on the basis of which the MIP-compatible model can be automatically generated. In this new format, the mode-switching behavior is incorporated using mixed-integer linear inequality constraints, requiring only $N_q$ Boolean variables.

\subsection{Input bounds}\label{mip:subsec:modelingsacss}

By \assref{mip:ass:modallydisjointinputs}, the SAcSS has disjoint inputs. Without loss of generality,
let us also assume the input channels in $u$ are ordered per mode according to
\begin{equation}\label{mip:eq:sacssu}
  \ifthesis
\else
  u = \begin{bmatrix} u^{1\top} & \dots & u^{N_q\top} \end{bmatrix}^\top \in \Us \subseteq \Rset^{n_u},
  \fi
\end{equation}
where $u^q$ represents the input of dimension $n_u^q=|\Phi(q)|$ corresponding to mode $q\in\Qs$. Note that $\sum_{q\in\Qs}n_u^q=n_u$. We define $\underline{u}^q\in\Rset^{n_u^q}$ and $\overline{u}^q\in\Rset^{n_u^q}$ as the corresponding lower and upper bounds, resulting in the input sets
\begin{equation}\label{mip:eq:inputboundsmodal}
  \Us_q = \{ u^q\in\Rset^{n_u^q} \mid \underline{u}^q \leq u^q \leq \overline{u}^q \},\quad q\in\Qs.
\end{equation}
We collect these bounds in the stacked vectors $\underline{u}=[\underline{u}^{1\top}\ \cdots\ \underline{u}^{N_q\top}]^\top$ and $\overline{u}=[\overline{u}^{1\top}\ \cdots\ \overline{u}^{N_q\top}]^\top$, yielding
\begin{equation}\label{mip:eq:inputboundsall}
  \Us = \{ u\in\Rset^{n_u} \mid \underline{u} \leq u \leq \overline{u} \}.
\end{equation}

\subsection{Setup time matrix}

We construct the \emph{setup time matrix}, being the weighted adjacency matrix of $\Gamma$, as
\begin{equation}\label{mip:eq:setuptimematrix}
  S \coloneqq \begin{bmatrix} s_{11} & \cdots & s_{1N_q} \\ \vdots & \ddots & \vdots \\ s_{N_q1} & \cdots & s_{N_qN_q} \end{bmatrix} \in \Nset^{N_q\times N_q},
\end{equation}
with $s_{q\tilde{q}}\coloneqq s(q,\tilde{q})$ for all $(q,\tilde{q})\in\Es$, and where we additionally define $s_{qq}=0$ for all $q\in\Qs$, reflecting the fact that staying in the same mode induces no setup time.

\subsection{Integer variables and admissibility conditions}\label{mip:subsec:inputconstraints}

\subsubsection{Activator}

To encode the actuator state $\sigma_k\in\Qs\cup\Es$, we use a one-hot Boolean vector (i.e., with one element equal to 1 and all others 0 \cite{Bemporad2011}), which we refer to as the \emph{activator}.
To this end, let $\Bs=\{0,1\}^{N_q}$ denote the set of $N_q$ Boolean variables, and consider the function $\Delta:\Qs\cup\Es\to\Bs$ defined as $\Delta(\sigma)=[\Delta_1(\sigma)\ \cdots\ \Delta_{N_q}(\sigma)]^\top$, where
\begin{subequations}\label{mip:eq:delta}
\begin{equation}\label{mip:eq:deltavars}
  \Delta_q(\sigma) = \left\{\begin{array}{ll} 1, &\text{if } \post(\sigma)=q, \\ 0, &\text{otherwise}. \end{array}\right.
\end{equation}
The activator $\delta_k$ at time $k$ is related to $\sigma_k$ according to
\begin{equation}\label{mip:eq:deltavec}
  \delta_k = \begin{bmatrix} \delta_k^1 \\ \vdots \\ \delta_k^{N_q} \end{bmatrix} = \begin{bmatrix} \Delta_1(\sigma_k) \\ \vdots \\ \Delta_{N_q}(\sigma_k) \end{bmatrix} = \Delta(\sigma_k) \in \Bs.
\end{equation}
The activator indicates which operational mode is (being) activated, i.e., $\delta_k^q=1$ if $\post(\sigma_k)=q$, meaning at time $k$ the SAcSS is either in the corresponding actuator mode ($\sigma_k=q$) or switching towards this mode ($\sigma_k=(\tilde{q},q)$), and $\delta_k^q=0$ otherwise. Clearly, it holds that the activator is a one-hot vector, and hence
\begin{equation}\label{mip:eq:deltasum}
  \sum_{q\in\Qs}\delta_k^q = 1.
\end{equation}
\end{subequations}

\begin{defn}\label{mip:def:admissibledelta}
  An activator sequence $\bm{\delta}_{[0,K]}=(\delta_0,\dots,\delta_K)\in\Bs^{K+1}$, $K\in\Nset$, is called \emph{$\Sigma$-admissible} if there is a $\Sigma$-admissible actuator sequence $\bm{\sigma}_{[0,K]}$ with $\delta_k=\Delta(\sigma_k)$ for all $k\in\Nset_{[0,K]}$.
\end{defn}

\begin{defn}\label{mip:def:admissibledeltau}
  A pair $(\bm{\delta}_{[0,K]},\bm{u}_{[0,K]})$, $K\in\Nset$, is called \emph{$\Sigma$-admissible} if there exists a $\Sigma$-admissible actuator sequence $\bm{\sigma}_{[0,K]}$ such that $\delta_k=\Delta(\sigma_k)$ for all $k\in\Nset_{[0,K]}$ (i.e., $\bm{\delta}_{[0,K]}$ is a $\Sigma$-admissible activator sequence) and the pair $(\bm{\sigma}_{[0,K]},\bm{u}_{[0,K]})$ is $\Sigma$-admissible.
\end{defn}

Note that while the actuator state $\sigma_k$ differentiates between the SAcSS having an operational actuator mode activated $(\sigma_k\in\Qs)$ or being in transition between two modes $(\sigma_k\in\Es)$, the activator $\delta_k$ effectively lumps these two actuator states together due to the function $\Delta$ considering only the destination mode $\post(\sigma_k)$. This property is instrumental in reducing the number of required Boolean variables in the model.
That is, explicitly modeling also the states $\sigma_k\in\Es$, as in a constrained switched (linear) system \cite{Liberzon2003,Philippe2016} in mixed logical dynamical (MLD) form \cite{Bemporad1999} for instance, would require additional Boolean variables. Moreover, since the switching is to be determined by MPC, where a performance-related stage cost must be evaluated at each prediction instant, switches $(q,\tilde{q})\in\Es$ with $s(q,\tilde{q})>1$ even require multiple additional Boolean variables to model the different levels of switching progress, which is referred to as lifting \cite{Subramanian2012}. As a result, such approaches require $\sum_{(q,\tilde{q})\in\Es}s_{q\tilde{q}}$ Boolean variables (which grows for increasing setup times) in addition to the $N_q$ Boolean variables that model the operational actuator states $\sigma_k\in\Qs$.
By contrast, our approach leads to a more compact model using only the $N_q$ Boolean variables in the activator, independent of the setup times. Since, however, the activator $\delta_K$ at some time $K\in\Nset$ does not indicate the actuator state $\sigma_K$ directly, see \eqref{mip:eq:delta}, we instead infer $\sigma_K$, and thereby the admissible input $u_K$ according to \defref{mip:def:admissiblesigmau}, from a \emph{sequence} of current and past activators in $\bm{\delta}_{[0,K]}$. Note that this essentially exchanges the need for many Boolean variables per time instant against the need for a few Boolean variables over multiple instants, and that an activator $\delta_k$, $k\in\Nset$, is required to determine $\sigma_l$ at multiple (equal and later) times $l\in\Nset_{\geq k}$. In the context of MPC, however, the decision variables such as the activators must be considered at all time instants in the prediction horizon anyway. Hence, including only the $N_q$-dimensional Boolean activator per time step, and effectively ``reusing'' each (predicted) activator to determine multiple (predicted) actuator states, leads to an MIP problem with significantly less Boolean variables, thereby improving efficiency.

\subsubsection{Admissibility conditions}

Next, we formulate explicit conditions on an input sequence $\bm{u}_{[0,K]}$ in relation to an activator sequence $\bm{\delta}_{[0,K]}$ by which the pair $(\bm{\delta}_{[0,K]},\bm{u}_{[0,K]})$ is admissible. To this end, however, we must
first introduce
\begin{equation}\label{mip:eq:setupmodeset}
  \Qs_{<\tau}^{q} = \{ \tilde{q}\in\Qs \mid s_{\tilde{q}q} < \tau \}, \quad \tau\in\Nset_{>0},\ q\in\Qs,
\end{equation}
denoting the set of modes $\tilde{q}\in\Qs$ from which the switch to mode $q\in\Qs$ induces a setup time smaller than $\tau$ instants. Note that $q\in\Qs_{<1}^q$ as $s_{qq}=0$ by definition, $\Qs_{<\tau_1}^q \subseteq \Qs_{<\tau_2}^q$ for positive times $\tau_1<\tau_2$, and $\Qs_{<\tau}^q=\Qs$ for all $\tau>\max_{\tilde{q}\in\Qs}s_{\tilde{q}q}$
(since $\Gamma$ is complete and satisfies the triangle inequality).
Using \eqref{mip:eq:setupmodeset}, we can
formulate the admissibility conditions as in the following theorem.

\begin{thm}\label{mip:thm:admissibledeltau}
  Given an admissible activator sequence $\bm{\delta}_{[0,K]}$, the pair $(\bm{\delta}_{[0,K]},\bm{u}_{[0,K]})$ is admissible if and only if for all $k\in\Nset_{[0,K]}$ it holds that $u_k^{\bar{q}}=0_{n_u^{\bar{q}}}$ for all $\bar{q}\in\Qs\setminus\{q\}$, where $q\in\Qs$ satisfies
  \begin{enumerate}[(I)]
    \item\label{mip:thm:conditionk} $\delta_k^{q}=1$, and,
    \item\label{mip:thm:conditionkmintau} for all $\tau\in\Nset_{[1,\min\{k,\overline{s}_q\}]}$,
    there exists $\tilde{q}\in\Qs_{<\tau}^{q}$ such that $\delta_{k-\tau}^{\tilde{q}}=1$.
  \end{enumerate}
\end{thm}

\begin{proof}
  By admissibility of $\bm{\delta}_{[0,K]}$ there exists a corresponding admissible actuator sequence $\bm{\sigma}_{[0,K]}$. For this $\bm{\sigma}_{[0,K]}$, the input sequence $\bm{u}_{[0,K]}$ is restricted exactly to its admissible set as in \defref{mip:def:admissiblesigmau} by $u_k^{\bar{q}}=0_{n_u^{\bar{q}}}$ for all $\bar{q}\in\Qs\setminus\{\sigma_k\}$ and all $k\in\Nset_{[0,K]}$. Hence, \thmref{mip:thm:admissibledeltau} can be proven by showing that the conditions~\eqref{mip:thm:conditionk}-\eqref{mip:thm:conditionkmintau} hold if and only if $\sigma_k=q$, as will be done in the remainder of this proof. To this end, first note that since $\Gamma$ is complete and satisfies the triangle inequality, we have that given some $\tilde{q}\in\Qs\setminus\Qs_{<\tau}^q$ there exists no admissible actuator sequence connecting $\tilde{q}$ to $q$ in less than $\tau$ steps.
  Consequently, inspecting which admissible past actuator sequences can lead to $\sigma_k=q$ can be done entirely on the basis of $\Qs_{<\tau}^q$, and thereby on only the actuator sequences describing ``direct'' switches to $q$, while disregarding all actuator sequences by which the SAcSS reaches $\sigma_k=q$ ``indirectly'' (i.e., by first switching to some other mode before finally switching to $q$).

  \emph{(Only if.)} We will first show that for an admissible $\bm{\sigma}_{[0,K]}$ and corresponding $\bm{\delta}_{[0,K]}$, the equality $\sigma_k=q$ implies conditions~\eqref{mip:thm:conditionk}-\eqref{mip:thm:conditionkmintau} in the theorem. By admissibility of $\bm{\delta}_{[0,K]}$, we can use \eqref{mip:eq:delta} to find
  \begin{equation}
    \sigma_k = q\ \Rightarrow\ \delta_k^q = \Delta_q(\sigma_k) = 1,
  \end{equation}
  and thus condition~\eqref{mip:thm:conditionk} is satisfied. Next, by admissibility of $\bm{\sigma}_{[0,K]}$, $\sigma_k=q$ implies that
  \begin{equation}\label{mip:eq:lemproofsigma}
    \sigma_{k-\tau} \in \Qs_{<\tau}^q \cup \{(\tilde{q},\bar{q})\in\Es \mid \bar{q}\in\Qs_{<\tau}^q\}
  \end{equation}
  for all $\tau\in\Nset_{[1,k]}$,
  i.e, the SAcSS must at all times $l=k-\tau$, $\tau\in\Nset_{[1,k]}$, have been in or transitioning towards a mode from which the switch to mode $q$ would induce a setup time of less than $\tau$ time instants (as otherwise the SAcSS would not be able to reach mode $q$ at time $k$). By admissibility of $\bm{\delta}_{[0,K]}$, we can apply the function $\Delta$ from \eqref{mip:eq:delta} to find that \eqref{mip:eq:lemproofsigma} implies that for all $\tau\in\Nset_{[1,k]}$ there exists $\tilde{q}\in\Qs_{<\tau}^q$ such that $\delta_{k-\tau}^{\tilde{q}}=1$,
  and thus condition~\eqref{mip:thm:conditionkmintau} holds.

  \emph{(If.)} By admissibility of $\bm{\delta}_{[0,K]}$, we use \eqref{mip:eq:delta} to derive from condition~\eqref{mip:thm:conditionk}
  \[
    \delta_k^q=1\ \Leftrightarrow\ \post(\sigma_k)=q\ \Leftrightarrow\ \sigma_k\in\{q\}\cup\{(\hat{q},q)\in\Es\},
  \]
  i.e., the SAcSS was either in or switching towards mode $q$ at time $k$. We will complete this proof by showing that condition~\eqref{mip:thm:conditionkmintau} excludes the latter possibility. To this end, note that for any admissible actuator sequence with $\sigma_k=(\hat{q},q)\in\Es$, it must hold that $\post(\sigma_{k-\tau})=\hat{q}\in\Qs\setminus\Qs_{<\tau}^q$ for some $\tau\in\Nset_{[1,\min\{k,s_{\hat{q}q}\}]}$, followed by $\sigma_l=(\hat{q},q)$ for all $l\in\Nset_{[k-\tau+1,k]}$ (describing that the SAcSS was in or arriving at mode $\hat{q}$ at time $k-\tau$, and that the switch $(\hat{q},q)$ started at time $k-\tau+1$ and is not yet completed at time $k$). Using \eqref{mip:eq:deltavec} by admissibility of $\bm{\delta}_{[0,K]}$, this is equivalent to $\delta_{k-\tau}^{\hat{q}}=1$ for some $\tau\in\Nset_{[1,\min\{k,s_{\hat{q}q}\}]}$ and $\hat{q}\notin\Qs_{<\tau}^q$ (and $\delta_l^q=1$ for all $l\in\Nset_{[k-\tau+1,k]}$).
  Using \eqref{mip:eq:deltasum}, however, we find that this is contradicted by condition~\eqref{mip:thm:conditionkmintau}, which completes the proof.
\end{proof}

To further clarify \thmref{mip:thm:admissibledeltau}, and in particular the first (only if) part of its proof, consider \exmref{mip:exm:equivalentsigmadelta}.

\begin{exmp}\label{mip:exm:equivalentsigmadelta}
  Consider a SAcSS with $\Gamma$ as in \exmref{mip:exm:sacssautomaton}. For an admissible actuator sequence with $\sigma_k=1$ for some $k\in\Nset_{[4,K]}$, the possible actuator states at the preceding time instants are depicted in \figref{mip:fig:automaton4_exm}, where the possible actuator states at each time are highlighted (black). If $\sigma_k=1$ (\figref{mip:fig:automaton4_exmK}), then at time $k-1$ by admissibility the SAcSS must have either already been in mode 1 or in the last time instant of switching towards this mode, i.e., $\sigma_{k-1}\in\{1\}\cup\{(q,1)\in\Es\}$ (\figref{mip:fig:automaton4_exmKm1}). Consequently, taking the setup times into account, at time $k-2$ the SAcSS may have been in modes 1 or 3 or in the last step of switching towards these modes, i.e., $\sigma_{k-1}\in\{1,3\}\cup\{(q,\tilde{q})\in\Es\mid\tilde{q}\in\{1,3\}\}$ (\figref{mip:fig:automaton4_exmKm2}), by which in turn at time $k-3$ (and all preceding times) the SAcSS may have been in any actuator state, i.e., $\sigma_{k-3}\in\Qs\cup\Es$ (\figref{mip:fig:automaton4_exmKm3}).
  \begin{figure}[!t]
  \begin{subfigure}[b]{0.5\textwidth}
    \centering
    \includegraphics[width=0.7\textwidth]{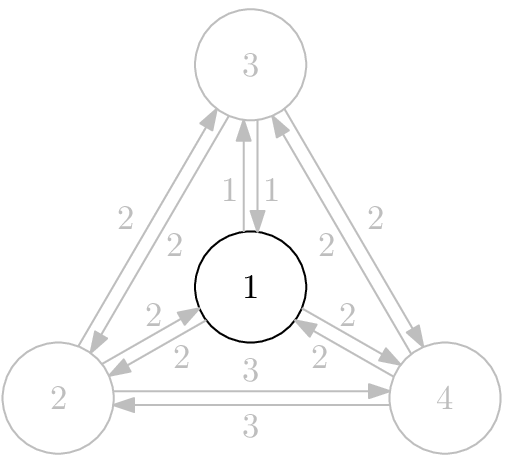}
    \caption{Time $k$: $\sigma_k=1$}
    \label{mip:fig:automaton4_exmK}
  \end{subfigure}
  \begin{subfigure}[b]{0.5\textwidth}
    \centering
    \includegraphics[width=0.7\textwidth]{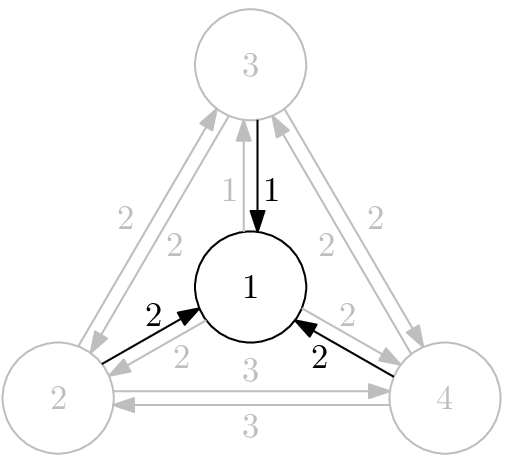}
    \caption{Time $k-1$: $\Qs_{<1}^1=\{1\}$}
    \label{mip:fig:automaton4_exmKm1}
  \end{subfigure}\\[1mm]
  \begin{subfigure}[b]{0.5\textwidth}
    \centering
    \includegraphics[width=0.7\textwidth]{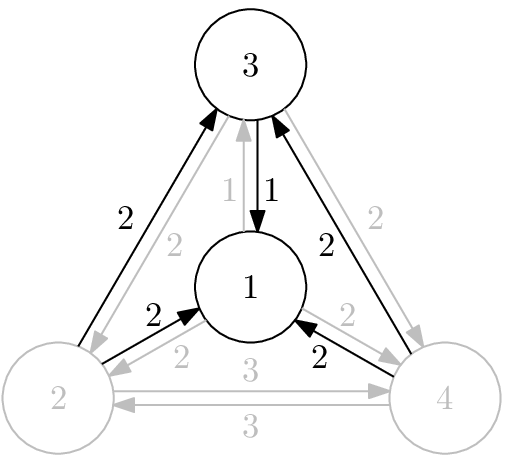}
    \caption{Time $k-2$: $\Qs_{<2}^1=\{1,3\}$}
    \label{mip:fig:automaton4_exmKm2}
  \end{subfigure}
  \begin{subfigure}[b]{0.5\textwidth}
    \centering
    \includegraphics[width=0.7\textwidth]{figures/automaton4_sacss.eps}
    \caption{Time $k-3$: $\Qs_{<3}^1=\Qs$}
    \label{mip:fig:automaton4_exmKm3}
  \end{subfigure}
  \caption{Possible actuator states (black) at different times, given an admissible actuator sequence with $\sigma_k=1$.}
  \label{mip:fig:automaton4_exm}
  \end{figure}%
  Now, realize that this in fact demonstrates \eqref{mip:eq:lemproofsigma}, which is the key result in the ``only if'' part of the proof of \thmref{mip:thm:admissibledeltau}.
\end{exmp}

\subsection{Mixed-integer linear inequality constraints}\label{mip:subsec:MIconstraints}

In this subsection, we use \thmref{mip:thm:admissibledeltau} to formulate the mixed-integer linear inequality constraints (in terms of the inputs and activators) that incorporate the actuator switching behavior in a manner that ensures admissibility of the pair $(\bm{\delta}_{[0,K]},\bm{u}_{[0,K]})$ for a \emph{given} $\bm{\delta}_{[0,K]}$. Then, we discuss a significant simplification of these constraints for systems with one-sided controls.

\subsubsection{General form}

Recall that $\underline{u}^q$ and $\overline{u}^q$ denote the lower and upper input bounds of the inputs in mode $q\in\Qs$ as in \eqref{mip:eq:inputboundsmodal}. For some $k\in\Nset$, the linear inequality constraints that ensure that $u_k^{\tilde{q}}=0_{n_u^{\tilde{q}}}$ for all $\tilde{q}\in\Qs\setminus\{q\}$, where $q\in\Qs$ satisfies condition~\eqref{mip:thm:conditionk} of \thmref{mip:thm:admissibledeltau}, are then given by
\begin{equation}\label{mip:eq:setuptimeconstraintkmode}
  \delta_k^{q}\underline{u}^q \leq u_k^q \leq \delta_k^{q}\overline{u}^q\quad \text{for all}\ q\in\Qs,
\end{equation}
indeed allowing for nonzero $u_k^q\in\Us_q$ only when $\delta_k^q=1$. Collecting all modal input bounds \eqref{mip:eq:setuptimeconstraintkmode} results in the constraints on $u_k$ corresponding to condition~\eqref{mip:thm:conditionk} to read
\ifthesis
\else
\begin{equation}\label{mip:eq:setuptimeconstraintk}
  \underline{U}_0\delta_k \leq u_k \leq \overline{U}_0\delta_k,
\end{equation}
where the $n_u\times N_q$-dimensional block diagonal input constraint matrices are given by $\underline{U}_0 = \diag(\underline{u}^1,\dots,\underline{u}^{N_q})$ and $\overline{U}_0 = \diag(\overline{u}^1,\dots,\overline{u}^{N_q})$.
\fi
Note that the products $\underline{U}_0\delta_k$ and $\overline{U}_0\delta_k$ contain mostly zeros, except for the entries corresponding to $\delta_k^q=1$.

By the same reasoning, the input restrictions resulting from condition~\eqref{mip:thm:conditionkmintau} of \thmref{mip:thm:admissibledeltau} for some $k\in\Nset$ can equivalently be expressed as requiring for all $q\in\Qs$ that
\begin{equation}\label{mip:eq:setuptimeconstraintkmintaumode}
  \sum_{\tilde{q}\in\Qs_{<\tau}^q}\delta_{k-\tau}^{\tilde{q}}\underline{u}^q \leq u_k^q \leq \sum_{\tilde{q}\in\Qs_{<\tau}^q}\delta_{k-\tau}^{\tilde{q}}\overline{u}^q\quad \text{for all}\ \tau\in\Nset_{[1,\min\{k,\overline{s}_q\}]}.
\end{equation}
To automate the generation of the inequality constraints \eqref{mip:eq:setuptimeconstraintkmintaumode} for the entire input vector, we first use the setup time matrix $S$ from \eqref{mip:eq:setuptimematrix} to construct the setup time constraint matrices
\begin{subequations}\label{mip:eq:Stau}
\begin{equation}
  S_\tau \coloneqq \begin{bmatrix} s_{\tau,11} & \cdots & s_{\tau,1N_q} \\ \vdots & \ddots & \vdots \\ s_{\tau,N_q1} & \cdots & s_{\tau,N_qN_q} \end{bmatrix} \in \{0,1\}^{N_q\times N_q}, \quad \tau=1,\ldots,\overline{s},
\end{equation}
with largest setup time $\overline{s}\coloneqq\max_{(q,\tilde{q})\in\Es}s_{q\tilde{q}}$, and
\begin{equation}
  s_{\tau,q\tilde{q}} \coloneqq \left\{\begin{array}{ll} 1, &\text{if } s_{q\tilde{q}}<\tau, \\ 0, &\text{otherwise}. \end{array}\right.
\end{equation}
\end{subequations}
The input constraints from \eqref{mip:thm:conditionkmintau} are then given by
\begin{subequations}\label{mip:eq:setuptimeconstraintkmintau}
\begin{equation}\label{mip:eq:setuptimeconstraintkmintauinequality}
  \underline{U}_\tau\delta_{k-\tau} \leq u_k \leq \overline{U}_\tau\delta_{k-\tau}\ \text{for all}\ \tau\in\Nset_{[1,\min\{k,\overline{s}\}]},
\end{equation}
with input constraint matrices
\begin{align}
  \underline{U}_\tau &= \begin{bmatrix} s_{\tau,11}\underline{u}^1 & \cdots & s_{\tau,N_q1}\underline{u}^1 \\ \vdots & \ddots & \vdots \\ s_{\tau,1N_q}\underline{u}^{N_q} & \cdots & s_{\tau,N_qN_q}\underline{u}^{N_q} \end{bmatrix}, \label{mip:eq:setuptimeconstraintkmintaumatricxfullmin} \\
  \overline{U}_\tau &= \begin{bmatrix} s_{\tau,11}\overline{u}^1 & \cdots & s_{\tau,N_q1}\overline{u}^1 \\ \vdots & \ddots & \vdots \\ s_{\tau,1N_q}\overline{u}^{N_q} & \cdots & s_{\tau,N_qN_q}\overline{u}^{N_q} \end{bmatrix}. \label{mip:eq:setuptimeconstraintkmintaumatricxfullmax}
\end{align}
\end{subequations}

In summary, \eqref{mip:eq:setuptimeconstraintk} and \eqref{mip:eq:setuptimeconstraintkmintau} translate the conditions in \thmref{mip:thm:admissibledeltau} into mixed-integer linear inequality constraints that are directly suitable for MI-MPC. Precisely, if we have an admissible pair $(\bm{\delta}_{[0,K]},\bm{u}_{[0,K]})$, then this pair satisfies the constraints in \eqref{mip:eq:setuptimeconstraintk} and \eqref{mip:eq:setuptimeconstraintkmintau}. Moreover, these constraints are constructed automatically on the basis of $\underline{u}$, $\overline{u}$, $\Phi$, and $S$, and result in $n_u(\overline{s}+1)$ (two-sided) linear inequalities.

\begin{rem}\label{mip:rem:redundantconstraints}
  In case of non-uniform maximum setup times per mode, i.e., if $\overline{s}_q\neq \overline{s}_{\tilde{q}}$ for some $q,\tilde{q}\in\Qs$ and thus $\overline{s}_q<\overline{s}$ for some $q\in\Qs$, the number of inequality constraints can be reduced by omitting the constraints that follow from all-ones columns in the matrices $S_\tau$, $\tau=1,\ldots,\overline{s}$. That is, these constraints are a byproduct of automatically constructing \eqref{mip:eq:setuptimeconstraintkmintaumode} in the form of \eqref{mip:eq:setuptimeconstraintkmintau}, and in fact correspond to \eqref{mip:eq:setuptimeconstraintkmintaumode} for $\tau>\overline{s}_q$, which are not included in condition~\eqref{mip:thm:conditionkmintau} (and which are satisfied by definition, see \remref{mip:rem:redundantconstraints}). In many numerical solvers, however, these constraints are automatically removed using a presolve algorithm, which eliminates the need to explicitly remove them during our modeling procedure.
\end{rem}

\begin{rem}
  If the system dynamics \eqref{mip:eq:sacss} are linear, then the combination of \eqref{mip:eq:sacss} with the linear inequality constraints \eqref{mip:eq:setuptimeconstraintk} and \eqref{mip:eq:setuptimeconstraintkmintau} in terms of continuous and Boolean variables in fact describes a SAcSS as an MLD system \cite{Bemporad1999}.
  Contrary to the general-purpose discrete hybrid automaton (DHA) framework (and accompanying modeling language HYSDEL which automates the conversion of a DHA into MLD form) \cite{Torrisi2004}, our method is tailored to SAcSSs. As a result, it allows for easier high-level specification of SAcSSs, and enables the automated generation of integer-wise compact (MLD) models.
\end{rem}

\subsubsection{Simplifications for one-sided control}\label{mip:subsubsec:constraintsonesided}

The number of input constraints in \eqref{mip:eq:setuptimeconstraintkmintau} may be reduced in case of one-sided input channels, i.e., when $\underline{u}_i=0$ or $\overline{u}_i=0$ for some $i\in\Nset_{[1,n_u]}$. This is the case in MR-HIFU hyperthermia, for example, where due to physical limitations heat cannot actively be removed from the system, and thus $u_{k,i}\geq\underline{u}_i=0$ for all $k\in\Nset$. For ease of exposition, we discuss here the simplification for the case that all inputs are nonnegative, i.e., $\underline{u}=0_{n_u}$. All other cases can be treated mutatis mutandis.

The key property that enables the reduction of the number of inequality constraints is that by the one-sidedness of the controls, the logic of condition~\eqref{mip:thm:conditionkmintau} can be imposed on the \emph{sum} of the one-sided inputs per mode, as opposed to on the individual input elements. That is, if $\underline{u}^q=0$ for all $q\in\Qs$, and additionally \eqref{mip:eq:setuptimeconstraintkmode} is imposed to restrict the individual input elements, then \eqref{mip:thm:conditionkmintau} (and thereby \eqref{mip:eq:setuptimeconstraintkmintaumode}) is equivalent to requiring for all $q\in\Qs$ that
\begin{equation*}
  1_{n_u^q}^\top u_k^q \leq \sum_{\tilde{q}\in\Qs_{<\tau}^q}\delta_{k-\tau}^{\tilde{q}} 1_{n_u^q}^\top \overline{u}^q\quad \text{for all}\ \tau\in\Nset_{[1,\min\{k,\overline{s}_q\}]}.
\end{equation*}
Consequently, we may replace \eqref{mip:eq:setuptimeconstraintkmintauinequality} by the inequalities
\ifthesis
\else
\begin{equation}\label{mip:eq:setuptimeconstraintkmintausum}
  J_{u}u_k \leq J_{u}\overline{U}_\tau\delta_{k-\tau}\quad \text{for all}\ \tau\in\Nset_{[1,\min\{k,\overline{s}\}]},
\end{equation}
with block diagonal matrix $J_{u} = \diag(1_{n_u^1}^\top,\dots,1_{n_u^{N_q}}^\top)\in\{0,1\}^{N_q\times n_u}$.
\fi
Together, \eqref{mip:eq:setuptimeconstraintk} and \eqref{mip:eq:setuptimeconstraintkmintausum} comprise the reduced set of inequality constraints that incorporate the SAcSS' switching including setup times.

\begin{rem}\label{mip:rem:setuptimeconstraintkmintausumuniform}
  Note that if, in addition, all modes have a common bound on the sum of their inputs, i.e., $1_{n_u^q}^\top u^q\leq\overline{u_\Sigma}$ with $\overline{u_\Sigma}\in\Rset_{>0}$ must hold for all $q\in\Qs$, then \eqref{mip:eq:setuptimeconstraintkmintausum} can be written as
  \begin{equation}\label{mip:eq:setuptimeconstraintkmintausumuniform}
    J_{u}u_k \leq \overline{u_\Sigma}S^\top_\tau\delta_{k-\tau}\quad \text{for all}\ \tau\in\Nset_{[1,\min\{k,\overline{s}\}]}.
  \end{equation}
\end{rem}

Summarizing, the $n_u\overline{s}$ two-sided inequality constraints in \eqref{mip:eq:setuptimeconstraintkmintauinequality} can be replaced by only $N_q\overline{s}$ single inequalities in \eqref{mip:eq:setuptimeconstraintkmintausum} (or \eqref{mip:eq:setuptimeconstraintkmintausumuniform}), achieving a reduction of $(2n_u-N_q)\overline{s}$ single inequalities. Especially for SAcSSs with high input dimension and only few modes, such as in MR-HIFU hyperthermia, and in the case of large maximum setup time $\overline{s}$, this leads to a significant reduction.

\begin{rem}\label{mip:rem:simplificationuminmax}
  In many numerical solvers, one needs to specify the upper and lower bounds of the search space for the optimization variables, by which the constraints $\underline{u}\leq u_k\leq\overline{u}$ are already included. Then, the above simplification procedure may also be applied to \eqref{mip:eq:setuptimeconstraintk}, yielding
  \begin{equation}\label{mip:eq:remsimplificationuminmax}
    J_uu_k \leq \overline{u_\Sigma}\delta_k.
  \end{equation}
  In this case, the $n_u(\overline{s}+1)$ two-sided inequalities in \eqref{mip:eq:setuptimeconstraintk} and \eqref{mip:eq:setuptimeconstraintkmintau} are reduced to the $N_q(\overline{s}+1)$ one-sided inequalities in \eqref{mip:eq:setuptimeconstraintkmintausum} and \eqref{mip:eq:remsimplificationuminmax}, resulting in a reduction of $(2n_u-N_q)(\overline{s}+1)$ constraints.
\end{rem}

\subsection{Admissibility assurance}

The inequalities \eqref{mip:eq:setuptimeconstraintk} and \eqref{mip:eq:setuptimeconstraintkmintau} (or, if applicable, \eqref{mip:eq:remsimplificationuminmax} and \eqref{mip:eq:setuptimeconstraintkmintausum} or \eqref{mip:eq:setuptimeconstraintkmintausumuniform}) provide a set of admissibility conditions on $\bm{u}_{[0,K]}$ in relation to $\bm{\delta}_{[0,K]}$. Being based on \thmref{mip:thm:admissibledeltau}, these constraints rely on the assumption that the activator sequence $\bm{\delta}_{[0,K]}$ is \emph{admissible}. To ensure the admissibility of $\bm{\delta}_{[0,K]}$ \emph{a priori} in the context of MPC, where the pair $(\bm{\delta}_{[0,K]},\bm{u}_{[0,K]})$ is to be determined by optimization, one could try to impose restrictions on $\bm{\delta}_{[0,K]}$. However, for compatibility with MI-MPC, this must then be done using additional mixed-integer linear inequality constraints, of which the derivation is not straightforward and due to linearization may require the introduction of many auxiliary integer variables. To avoid the need for this, we instead opt to ensure the admissibility of $\bm{\delta}_{[0,K]}$ \emph{a posteriori}. In particular, we will use only \eqref{mip:eq:setuptimeconstraintk} and \eqref{mip:eq:setuptimeconstraintkmintau} directly in the MPC to obtain a pair $(\bm{\tilde{\delta}}_{[0,K]},\bm{u}_{[0,K]})$ that satisfies the conditions in \thmref{mip:thm:admissibledeltau}, where $\bm{\tilde{\delta}}_{[0,K]}$ need not be admissible. Subsequently, on the basis of $(\bm{\tilde{\delta}}_{[0,K]},\bm{u}_{[0,K]})$ we construct (using a simple procedure) a potentially different activator sequence $\bm{\delta}_{[0,K]}$ such that $(\bm{\delta}_{[0,K]},\bm{u}_{[0,K]})$ is admissible for the same $\bm{u}_{[0,K]}$. We call any such pair $(\bm{\tilde{\delta}}_{[0,K]},\bm{u}_{[0,K]})$ \emph{$\Sigma$-feasible} (or \emph{feasible} if $\Sigma$ is clear from context), as defined next.

\begin{defn}\label{mip:def:feasibledeltau}
  A pair $(\bm{\delta}_{[0,K]},\bm{u}_{[0,K]})$, $K\in\Nset$, is called \emph{$\Sigma$-feasible} if it satisfies the conditions on $\bm{u}_{[0,K]}$ in relation to $\bm{\delta}_{[0,K]}$ given in \thmref{mip:thm:admissibledeltau}.
\end{defn}

The construction of an alternative sequence $\bm{\delta}_{[0,K]}$ such that the pair $(\bm{\delta}_{[0,K]},\bm{u}_{[0,K]})$ is $\Sigma$-admissible, given a $\Sigma$-feasible pair $(\bm{\tilde{\delta}}_{[0,K]},\bm{u}_{[0,K]})$, is based on the following theorem.

\begin{thm}\label{mip:thm:admissassurance}
  Let a $\Sigma$-feasible pair $(\bm{\tilde{\delta}}_{[0,K]},\bm{u}_{[0,K]})$ be given. Define $\Ks = \{k_1,\ldots,k_L\} = \{k\in\Nset_{[0,K]} \mid u_k\neq0\}$, $L\in\Nset$, where $k_l<k_{l+1}$ for all $l\in\Nset_{[1,L-1]}$, and construct the activator sequence $\bm{\delta}_{[0,K]}$ as
  \begin{equation}\label{mip:eq:thmadmissassurance}
    \delta_k = \left\{
               \begin{array}{ll}
                 \tilde{\delta}_{k_1},  & \text{if}\ k\in\Nset_{[0,k_1]}, \\
                 \tilde{\delta}_{k_l},  & \text{if}\ k\in\Nset_{[k_{l-1}+1,k_l]},\ l\in\Nset_{[2,L]}, \\
                 \tilde{\delta}_{k_L},  & \text{if}\ k\in\Nset_{[k_L+1,K]}.
               \end{array}
               \right.
  \end{equation}
  Then, the pair $(\bm{\delta}_{[0,K]},\bm{u}_{[0,K]})$ is $\Sigma$-admissible.
\end{thm}

\begin{proof}
  Note that $\Ks$ denotes the chronologically ordered set of $L$ time instants at which the control inputs are nonzero. Let $q_l$ with $l\in\Nset_{[1,L]}$ denote the mode with nonzero input at time $k_l\in\Ks$, i.e., such that $u_{k_l}^{q_l}\neq0$. By feasibility of $(\bm{\tilde{\delta}}_{[0,K]},\bm{u}_{[0,K]})$, condition~\eqref{mip:thm:conditionkmintau} of \thmref{mip:thm:admissibledeltau} holds for all $k\in\Nset_{[0,K]}$. Combined with \eqref{mip:eq:deltasum}, this implies that $q_l\in\Qs_{<k_{l+1}-k_l}^{q_{l+1}}$ for all $l\in\Nset_{[1,L-1]}$, and therefore
  \begin{equation}\label{mip:eq:thmproofadmissassurancesetuptime}
    s_{q_lq_{l+1}} < k_{l+1}-k_l \quad \text{for all}\ l\in\Nset_{[1,L-1]},
  \end{equation}
  by which we can construct the actuator sequence $\bm{\sigma}_{[0,K]}$ as
  \begin{equation*}
    \sigma_k = \left\{
               \begin{array}{ll}
                 q_1,           & \text{if } k\in\Nset_{[0,k_1]}, \\
                 (q_l,q_{l+1}), & \text{if } k\in\Nset_{[k_l+1,k_l+s_{q_lq_{l+1}}]},\ l\in\Nset_{[1,L-1]}, \\
                 q_l,           & \text{if } k\in\Nset_{[k_{l-1}+s_{q_{l-1}q_l}+1,k_l]},\ l\in\Nset_{[2,L]}, \\
                 q_L,           & \text{if } k\in\Nset_{[k_L,K]}.
               \end{array}
               \right.
  \end{equation*}
  First, note that this sequence $\bm{\sigma}_{[0,K]}$ is admissible according to \defref{mip:def:admissiblesigma}, describing the SAcSS (a) starting and staying in mode $q_1\in\Qs$ during the times $k\in\Nset_{[0,k_1]}$, (b) performing each switch $(q_l,q_{l+1})$, $l\in\Nset_{[1,L-1]}$, during the interval $k\in\Nset_{[k_l+1,k_l+s(q_l,q_{l+1})]}$ in case $q_l\neq q_{l+1}$, or remaining in the same mode during the interval $k\in\Nset_{[k_l,k_{l+1}]}$ if $q_l=q_{l+1}$ (since then $s_{q_lq_{l+1}}=0$ in the above construction of $\bm{\sigma}_{[0,K]}$), and (c) staying in the last active mode $q_L$ during $k\in\Nset_{[k_L,K]}$. Since, additionally, $\sigma_{k_l}=q_l$ for all $l\in\Nset_{[1,L]}$, the pair $(\bm{\sigma}_{[0,K]},\bm{u}_{[0,K]})$ is admissible according to \defref{mip:def:admissiblesigmau}. Finally, note that for $\bm{\delta}_{[0,K]}$ from \eqref{mip:eq:thmadmissassurance} it holds that $\delta_k=\Delta(\sigma_k)$ for all $k\in\Nset_{[0,K]}$, which proves the admissibility of the pair $(\bm{\delta}_{[0,K]},\bm{u}_{[0,K]})$ according to \defref{mip:def:admissibledeltau}.
\end{proof}

We call the procedure in \thmref{mip:thm:admissassurance} the \emph{admissibility assurance}. Note that this procedure is simple, and thus computationally lightweight, and can be performed automatically. The activator sequence constructed using the admissibility assurance corresponds to the actuator sequence that enables $\bm{u}_{[0,K]}$ in an admissible manner, using the minimum number of mode switches, and where each switch occurs as early as possible.

\ifthesis
\else
\section{Mixed-Integer MPC for SAcSSs}\label{mip:sec:mipmpcsacss}
\fi

In this section, we first show how a SAcSS model derived using the procedure discussed in the previous section can be directly integrated into MI-MPC. Then, we present the control algorithm, taking into account that re-optimizing during a switch is redundant, as all switches must be completed by definition.

\subsection{Optimal control problem}

For the purpose of MPC, let us first introduce the sequences $\bm{x}_k=(x_{0|k},\ldots,x_{N|k})$, $\bm{u}_k=(u_{0|k},\ldots,u_{N-1|k})$, and $\bm{\delta}_k=(\delta_{0|k},\ldots,\delta_{N-1|k})$, where $x_{i|k}\in\Xs$, $u_{i|k}\in\Us$, and $\delta_{i|k}\in\Bs$ denote the predicted states, inputs, and activators at $i\in\Nset$ time instants ahead of the prediction's starting time $k\in\Nset$. Then, the SAcSS MI-MPC optimization problem is given by
\begin{subequations}\label{mip:eq:OCP}
\begin{equation}\label{mip:eq:OCPcost}
  \min_{\bm{u}_k,\bm{\delta}_k} V_N(\bm{x}_k,\bm{u}_k)
\end{equation}
subject to the set constraints $\bm{u}_k\in\bm{\Us}_N(\bm{x}_k)$ and $\bm{\delta}_k\in\Bs^N$, the system dynamics \eqref{mip:eq:sacss} by $x_{i+1|k}=f(x_{i|k},u_{i|k})$ for all $i\in\Nset_{<N}$ with initial condition $x_{0|k}=x_k$, and
\begin{equation}\label{mip:eq:OCPdeltasum}
  \begin{bmatrix}[c@{\enskip}|@{\enskip}cc]
    0_{N\times Nn_u} & O_1(k)   & I_N\otimes1_{N_q}^\top
  \end{bmatrix}
  \left[\begin{array}{@{}c@{}} \mathbf{u}_k \\\hline \mathbf{d}_k \end{array}\right]
  = 1_N,
\end{equation}
\begin{equation}\label{mip:eq:OCPsetuptimes}
  \begin{bmatrix}[c@{\enskip}|@{\enskip}ccc]
    -I_{Nn_u}   & O_2(k,\tau) & I_{N-\tau}\otimes\underline{U}_\tau   & O_3(\tau) \\[1mm]
    I_{Nn_u}    & O_2(k,\tau) & -I_{N-\tau}\otimes\overline{U}_\tau   & O_3(\tau)
  \end{bmatrix}
  \left[\begin{array}{@{}c@{}} \mathbf{u}_k \\\hline \mathbf{d}_k \end{array}\right]
  \leq 0_{2Nn_u},
\end{equation}
\end{subequations}
for all $\tau\in\Nset_{[0,\overline{s}]}$. Here, $V_N:\Xs^{N+1}\times\Us^N\to\Rset_{\geq0}$ is the cost function over prediction horizon $N\in\Nset_{>0}$ (typically set to $N\in\Nset_{>\overline{s}}$, see \remref{mip:rem:horizon}), which incorporates the control objective via a penalty on the predicted states $\bm{x}_k$ and inputs $\bm{u}_k$, and can be chosen in different ways (e.g., linear or quadratic) following the MPC literature \cite{Borrelli2017,Bemporad2011,Rawlings2017a}. Next, $\bm{\Us}_N(\bm{x}_k)\subseteq\Us^N$ represents the set of input sequences by which the predicted states $\bm{x}_k$ and inputs $\bm{u}_k$ satisfy their corresponding constraints. Hence, the set membership $\bm{u}_k\in\bm{\Us}_N(\bm{x}_k)$ imposes all general state and input constraints. For the switching constraints specific to our SAcSSs model, \eqref{mip:eq:OCPdeltasum} incorporates \eqref{mip:eq:deltasum}, and \eqref{mip:eq:OCPsetuptimes} corresponds to \eqref{mip:eq:setuptimeconstraintk} and \eqref{mip:eq:setuptimeconstraintkmintau}. Here, $\mathbf{u}_k=[u_{0|k}^\top\ \dots\ u_{N-1|k}^\top]^\top$ denotes the stacked vector containing all elements of $\bm{u}_k$, and
$\mathbf{d}_k=[\delta_{-\min\{k,\overline{s}\}|k}^\top\ \dots\ \delta_{N-1|k}^\top]^\top$, where $\delta_{k-\tau|k}=\delta_{k-\tau}$ for $\tau\in\Nset_{>0}$,
contains the $\min\{k,\overline{s}\}$ activators prior to time $k$ and all activators in $\bm{\delta}_k$. Finally, $\otimes$ denotes the Kronecker product, and for ease of presentation we use the zero matrix shorthand notations $O_1(k)=0_{N\times N_q\min\{k,\overline{s}\}}$, $O_2(k,\tau)=0_{Nn_u\times N_q\min\{k,\overline{s}-\tau\}}$ and $O_3(\tau)=0_{Nn_u\times N_q\tau}$.

\begin{rem}\label{mip:rem:horizon}
  The SAcSS MI-MPC \eqref{mip:eq:OCP} can generally be executed for any $N\in\Nset_{>0}$. However, $N\in\Nset_{>\overline{s}}$ may often be desired, since otherwise the MPC is unable to take into account the predicted effect on the state resulting from initiating any switch $(q,\tilde{q})\in\Es$ with $s_{q\tilde{q}}\geq N$ and then applying nonzero inputs in the destination mode $\tilde{q}$.
\end{rem}

Thus, due to the structure in \eqref{mip:eq:OCPdeltasum}-\eqref{mip:eq:OCPsetuptimes}, the constraints \eqref{mip:eq:deltasum}, \eqref{mip:eq:setuptimeconstraintk}, and \eqref{mip:eq:setuptimeconstraintkmintau} that describe the mode switching and setup times can directly and automatically be used to formulate MI-MPCs for SAcSSs, which can be solved using off-the-shelf numerical optimizers. Moreover, \eqref{mip:eq:OCPdeltasum}-\eqref{mip:eq:OCPsetuptimes} are linear in terms of the activators and inputs. Consequently, in case $V_N$ is quadratic and $\bm{u}_k\in\bm{\Us}_N(\bm{x}_k)$ translates into linear constraints, \eqref{mip:eq:OCP} is an MIQP, for which efficient solvers are available (e.g., Gurobi, CPLEX, MOSEK).

Next, note that by \defref{mip:def:admissibledeltau} and \defref{mip:def:feasibledeltau} the admissibility of a pair $(\bm{\delta}_k,\bm{u}_k)$ implies its feasibility. The converse does not hold, but given a feasible pair $(\bm{\tilde{\delta}}_k,\bm{u}_k)$ the existence of an admissible pair with the same input sequence is guaranteed (and it can be constructed using \thmref{mip:thm:admissassurance}), from which we conclude the set of admissible pairs to be contained in the set of feasible pairs. Consequently, since in solving \eqref{mip:eq:OCP} to find an optimal pair $(\bm{\tilde{\delta}}_k^*,\bm{u}_k^*)$ we minimize over the set of feasible solutions, the optimal objective function value $V_N^*=V_N(\bm{x}_k,\bm{u}_k^*)$ is a lower bound for the minimization performed over the admissible solutions only. However, now note that in $V_N(\bm{x}_k,\bm{u}_k)$ \eqref{mip:eq:OCPcost} we assume no explicit cost on the predicted activator sequence $\bm{\delta}_k$, which is valid for many applications where the control-related performance measure is mainly based on the system states and the required control effort. As a result, applying the admissibility assurance to a feasible pair $(\bm{\tilde{\delta}}_k^*,\bm{u}_k^*)$ that optimizes \eqref{mip:eq:OCP} yields an admissible pair $(\bm{\delta}_k^*,\bm{u}_k^*)$
that also optimizes \eqref{mip:eq:OCP} with the same objective function value, and thus truly minimizes the optimal control problem evaluated over the set of admissible solutions.

\subsection{MI-MPC algorithm}

Recall that if a mode switch is initiated, for admissibility it must also be completed (see \eqref{mip:def:admissiblesigma2} of \defref{mip:def:admissiblesigma}), and that during a switch all control inputs are zero (see \defref{mip:def:admissiblesigmau}). Therefore, in the context of MPC there is no need for re-optimizing the control inputs and activators while the SAcSS is in the process of switching. When incorporating this practical consideration, the proposed MI-MPC setup is given by \algoref{mip:alg:MIMPC}. The variable $s_t$ stores the setup time of the last/current switch, and the counter $k_t$ represents the number of elapsed time steps since starting this switch. Initially, we set $k=k_t=s_t=0$, perform an admissibility-assured MI-MPC optimization by consecutively solving \eqref{mip:eq:OCP} and \eqref{mip:eq:thmadmissassurance}, and apply the optimal pair $(\delta_{0|k}^*,u_{0|k}^*)$ to the system. If at some time $k-1$ the SAcSS was in mode $q\in\Qs$, and at time $k$ it is optimal to start transitioning towards another mode $\tilde{q}\in\Qs$, $q\neq\tilde{q}$, indicated by $\delta_k\neq\delta_{k-1}$, $s_t$ stores the setup time $s_{q\tilde{q}}$, and $k_t$ is reset to zero. At the following time steps, we iteratively increase $k_t$ by one, using which we can evaluate whether the setup time required for the current transition has elapsed. If so, i.e., $k_t\geq s_t$, we perform another admissibility-assured MI-MPC optimization by consecutively solving \eqref{mip:eq:OCP} and \eqref{mip:eq:thmadmissassurance}. If not, i.e., $k_t<s_t$, we simply continue the mode switch ($\delta_k=\delta_{k-k_t}$) while disabling all inputs ($u_k=0_{n_u}$) until this transition has completed.
\begin{algorithm}[t]
  \caption{MI-MPC for for SAcSSs}\label{mip:alg:MIMPC}
  \begin{algorithmic}
    \Require SAcSS-based MI-MPC \eqref{mip:eq:OCP}, and initial conditions including $k=k_t=s_t=0$
    \While{SAcSS feedback control}
      \If{$k_t\geq s_t$}
        \State Solve \eqref{mip:eq:OCP} and \eqref{mip:eq:thmadmissassurance} to find  $(\bm{\delta}_k^*,\bm{u}_k^*)$
        \State Apply $(\delta_k,u_k) \gets (\delta_{0|k}^*,u_{0|k}^*)$
        \If{$\delta_k\neq\delta_{k-1}$}
          \State $k_t \gets 0$
          \State \parbox[t]{.8\linewidth}{$s_t \gets s(q,\tilde{q})$, where $q,\tilde{q}\in\Qs$ correspond to\\$\delta_{k-1}^q=1$ and $\delta_k^{\tilde{q}}=1$}
        \EndIf
      \Else
        \State Apply $(\delta_k,u_k) \gets (\delta_{k-k_t},0_{n_u})$
      \EndIf
      \State $k \gets k+1$
      \State $k_t \gets k_t+1$
    \EndWhile
  \end{algorithmic}
\end{algorithm}%

\ifthesis
\else
\section{Case Study}\label{mip:sec:mrhifu}
\fi

This section first discusses the MR-HIFU hyperthermia setup and treatment. Then, we apply the procedure proposed in \secref{mip:sec:modeling} to derive a compact SAcSS model, set up the MI-MPC, and perform validating simulations.

\subsection{MR-HIFU hyperthermia therapy platform}

The MR-HIFU system considered in this case study consists of a Profound Sonalleve HIFU platform,
and a Philips 3T Achieva MRI scanner. The former is a dedicated trolley-tabletop in which an MR-compatible HIFU transducer and its carrier system are integrated. The latter noninvasively provides near-real-time temperature measurements. This setup is already being used in clinics for the treatment of uterine fibroids, desmoid tumors, osteoid osteomas, and painful bone metastases.

\subsubsection{MR thermometry}\label{mip:subsubsec:MRT}

The temperature maps are obtained noninvasively using the proton resonance frequency shift (PRFS) method \cite{Ishihara1995,Winter2016}, providing the temperature difference with respect to some baseline by comparing the current MR image to a reference image. The reference image is typically acquired before treatment, such that the baseline corresponds to zero treatment-induced temperature elevation.
Due to the relative nature of this MR-based thermometry, combined with the fact that the HIFU transducer itself distorts the magnetic field, for accurate measurements given a certain transducer position, a baseline image is required that was obtained with the transducer in the same position. Therefore, we constrain the possible mechanical transducer positions to a discrete set, such that we can obtain a reference image for each location before treatment, and use a lookup table to select the correct reference image during treatment.

\subsubsection{HIFU transducer}

The Sonalleve contains a phased-array transducer consisting of 256 acoustic elements, each of which is able to generate high-intensity ultrasound waves. By coordinated modulation of the individual elements' phases and amplitudes, referred to as electronic beam steering, a focal spot can be created, see \figref{mip:fig:mrhifupatient}, and steered through a circular area with a diameter of 16~mm, which we call a treatment cell.
\begin{figure}[t]
  \centering
  \includegraphics[width=0.8\textwidth]{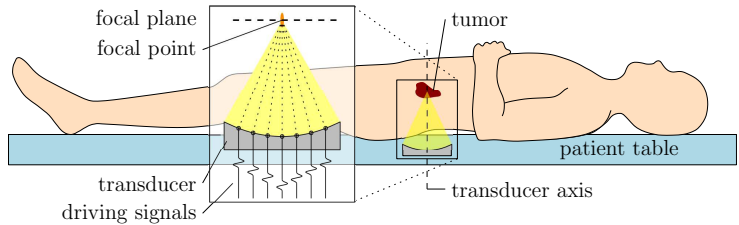}
  \caption{Schematic of a focal point by electronic beam steering, with the focal plane in the tumor volume.}
  \label{mip:fig:mrhifupatient}
\end{figure}%
To treat larger regions, the transducer itself must be mechanically relocated, see \cite{Tillander2016}, enabling the heating of multiple treatment cells throughout the treatment area.
Due to the previously mentioned limitations of the MR thermometry, the treatment cell locations are limited to a predefined discrete set. In this case study, we allocate $N_q=4$ treatment cells as depicted in \figref{mip:fig:corner4lumpin}, which we will explain further below.
\begin{figure}[t]
  \centering
  \psfrag{cell1}[cc][cc][0.9][0]{$q=1$}
  \psfrag{cell2}[cc][cc][0.9][0]{$q=2$}
  \psfrag{cell3}[cc][cc][0.9][0]{$q=3$}
  \psfrag{cell4}[cc][cc][0.9][0]{$q=4$}
  \psfrag{cellr}[cc][cc][0.9][0]{$\Rs$}
  \psfrag{cells}[cc][cc][0.9][0]{$\Ss$}
  \psfrag{x}[cc][cc][0.9][0]{$r_x$ [m]}
  \psfrag{y}[cc][cc][0.9][0]{$r_y$ [m]}
  \psfrag{-0.02}[cc][cc][0.9][0]{-0.02}
  \psfrag{-0.01}[cc][cc][0.9][0]{ -0.01}
  \psfrag{0}[cc][cc][0.9][0]{0}
  \psfrag{0.01}[cc][cc][0.9][0]{ 0.01}
  \psfrag{0.02}[cc][cc][0.9][0]{0.02}
  \ifthesis
\else
  \includegraphics[scale=0.80,clip]{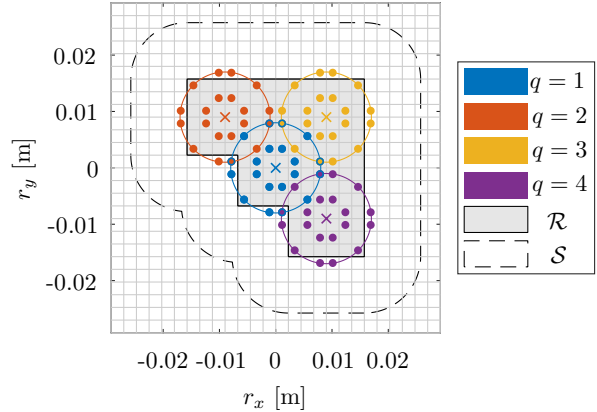}
  \fi
  \caption{The centers ($\times$), electronic beam steering ranges (circles), and sonication points ($\bullet$) of the four treatment cells, over a large ROI $\Rs$. The tissue outside $\Ss$ must be safeguarded against overheating.}
  \label{mip:fig:corner4lumpin}
\end{figure}%

\subsection{Hyperthermia treatment}

The main goal in a hyperthermia treatment is to create and maintain a controlled and homogeneous temperature elevation in the region of interest (ROI) $\Rs$ shown in \figref{mip:fig:corner4lumpin} containing the tumor, as this sensitizes the treated tissue to the effects of chemo- and radiotherapy. The tissue sensitization gradually increases for rising temperatures starting around 40~$^\circ$C \cite{Issels2008}, which is why the healthy tissue outside $\Ss$ in \figref{mip:fig:corner4lumpin} will be safeguarded against temperatures above 40~$^\circ$C. In the ROI, adequate sensitization is considered to occur around 41~$^\circ$C, and optimal treatment quality is achieved at 42~$^\circ$C, while overheating above 43~$^\circ$C reduces the beneficial heat-induced effects and should therefore be avoided.

\subsection{SAcSS model}

\subsubsection{Arc-weighted digraph}

The complete simple arc-weighted digraph $\Gamma=(\Qs,\Es,s)$ describing the system considered in this case study corresponds to \exmref{mip:exm:sacssautomaton}. Here, the setup time defined by $s$ represents the number of samples required to mechanically relocate the transducer from one cell to another.

\subsubsection{Plant dynamics, state space, input space, and actuator selector function}\label{mip:subsubsec:sacss}

To model the tissue's thermal dynamics, we follow a procedure similar to
\ifthesis
\else
\cite{Deenen2020a,Deenen2020}.
\fi
In summary, this entails first describing the tissue's thermal dynamics using the Pennes bioheat equation \cite{Pennes1948} using tissue parameters from \cite{ITISFoundation2018} and then spatially discretizing this partial differential equation on a two-dimensional square grid with $2.25\times2.25$~mm$^2$ voxels (depicted in \figref{mip:fig:corner4lumpin} by the grid lines) using the central difference scheme, which in \cite{Sebeke2019} was verified to be the method that best balances model simplicity with descriptive accuracy. Finally, the model is temporally discretized with the MR thermometry sample time $T_s=3.2$~s using the forward Euler method, which preserves model sparsity, while providing sufficient accuracy for the considered system%
\ifthesis
\else
\cite{Deenen2020}.
\fi

The resulting discrete-time state-space model in the form of \eqref{mip:eq:sacss} is given by
\begin{subequations}\label{mip:eq:model}
\begin{align}
  x_{k+1} &= f(x_k,u_k) = Ax_{k} + Bu_{k}, \label{mip:eq:modelx} \\
  y_k &= \left\{\begin{array}{ll} x_k + v_k, & \text{if}\ \sigma_{k-1}\in\Qs, \\ \emptyset, & \text{if}\ \sigma_{k-1}\in\Es, \end{array}\right. \label{mip:eq:modely}
\end{align}
\end{subequations}
where the states $x_k\in\Xs=\Rset^{n_x}$, $n_x=36^2=1296$, represent the temperature elevations with respect to the baseline of the voxels in the focal plane at real time $t_k=kT_s$. The matrix $A$ captures the effects of heat transfer by conduction and by blood perfusion in the form of first-order thermal dynamics and heat dissipation, respectively. As a result, $A$ is a Schur matrix, i.e., with all eigenvalues strictly inside the unit disc of the complex plane. The voxels are chosen such that their centers coincide with the points measured by MR thermometry, resulting in $y_k\in\Rset^{n_x}$ to consist of full state measurements corrupted by MR measurement noise $v_k\in\Rset^{n_x}$, which can be well approximated by spatially uncorrelated zero-mean Gaussian noise with a standard deviation of $0.4$~$^\circ$C, i.e., $v_k\sim\Ns(0,0.4^2I_{n_x})$, when the transducer was not moving during the interval from $t_{k-1}$ until $t_k$ indicated by $\sigma_{k-1}\in\Qs$. In case of transducer motion $\sigma_{k-1}\in\Es$, no measurement is available at time $t_k$. Per treatment cell, and thereby per actuator mode $q\in\Qs$, we use $n_u^q=n_{u,\mathrm{cell}}=20$ voxels at the centers of which we allow sonication by appropriate steering of the focal spot. These locations are referred to as sonication points, which are also shown in \figref{mip:fig:corner4lumpin} by the markers `$\bullet$'. Using our method, we only need to model the $N_q$ \emph{operational} input modes. To this end, let the input $u_k^q$ describe the average acoustic power applied at each of the sonication points in cell $q\in\Qs$ over the course of one sampling interval. The inputs are collected in the input vector $u_k\in\Us$ as in \eqref{mip:eq:sacssu}, and hence the input matrix is given by
\begin{equation}\label{mip:eq:modelB}
  B = \begin{bmatrix} B^1 & \cdots & B^{N_q} \end{bmatrix}\in\Rset^{n_x\times n_u},
\end{equation}
where each submatrix $B^q\in\Rset^{n_x\times n_{u,\mathrm{cell}}}$ describes the system's temperature change in response to the heating power $u_k^q$ applied at the sonication points in cell $q\in\Qs$. Since by HIFU we can only deposit heat in the system, but not extract it, the inputs are nonnegative. Additionally, for patient safety, the maximum input power per sonication point is limited to $u_{\max}=15$~W. Consequently, the input space is given by $\Us=\Rset_{[0,u_{\max}]}^{n_u}$, $n_u=\sum_{q\in\Qs}n_u^q=N_qn_{u,\mathrm{cell}}=80$. Next, note that even though six sonication points on the outer ring of cell 1 are shared by multiple cells, to satisfy \assref{mip:ass:modallydisjointinputs} we have defined the $n_{u,\mathrm{cell}}$ sonication points for each cell individually, resulting in disjoint inputs per actuator mode. Correspondingly,
the inputs are ordered as in \eqref{mip:eq:sacssu}, and the actuator selector function is given by
\begin{equation}\label{mip:eq:modelPhi}
  \Phi(\sigma) = \left\{\begin{array}{ll} \Nset_{[1+n_{u,\mathrm{cell}}(q-1),n_{u,\mathrm{cell}}q]}, & \text{if }\sigma=q\in\Qs, \\ \emptyset, & \text{if }\sigma\in\Es, \end{array}\right.
\end{equation}
describing the fact that heating can only occur in the cell where the transducer is located, and that no heating may occur during transducer motion.

\subsection{Input constraints}

In this case study, we use Gurobi 8.1.1, which requires explicit specification of the input's upper and lower limits. From $\Us$ as in \secref{mip:subsubsec:sacss}, according to \eqref{mip:eq:inputboundsall} we find
\begin{equation}\label{mip:eq:modelinputconstraintsu}
  0_{n_u} = \underline{u} \leq u_k \leq \overline{u} = 1_{n_u}u_{\max}.
\end{equation}
In addition, as a safety measure in the considered MR-HIFU hyperthermia setup we constrain the total applied power in a single treatment cell to 100~W, i.e.,
\begin{equation}\label{mip:eq:modelinputconstraintsusum}
  1_{n_u^q}^\top u^q \leq \overline{u_\Sigma} = 100 \quad \text{for all}\ q\in\Qs.
\end{equation}
To formulate the inequality constraints that incorporate the switching and setup times in the upcoming MI-MPC, of which the general form is given by \eqref{mip:eq:setuptimeconstraintk} and \eqref{mip:eq:setuptimeconstraintkmintau}, based on $\Gamma$ we first construct the setup time matrix \eqref{mip:eq:setuptimematrix}
\ifthesis
\else
and the corresponding constraint matrices \eqref{mip:eq:Stau}, which read
\begin{equation}\label{mip:eq:modelStau}
  S = \begin{bmatrix} 0 & 2 & 1 & 2 \\ 2 & 0 & 2 & 3 \\ 1 & 2 & 0 & 2 \\ 2 & 3 & 2 & 0 \end{bmatrix},\
  \left\{
  \begin{array}{c}
  S_1 = I_4,\\
  S_2 = \begin{bmatrix} 1 & 0 & 1 & 0 \\ 0 & 1 & 0 & 0 \\ 1 & 0 & 1 & 0 \\ 0 & 0 & 0 & 1 \end{bmatrix},\
  S_3 = \begin{bmatrix} 1 & 1 & 1 & 1 \\ 1 & 1 & 1 & 0 \\ 1 & 1 & 1 & 1 \\ 1 & 0 & 1 & 1 \end{bmatrix}.
  \end{array}
  \right.
\end{equation}
\fi
However, since by $\Us\subset\Rset_{\geq0}$ the controls are nonnegative, see also \eqref{mip:eq:modelinputconstraintsu}, and a common upper bound in the sum of the inputs is imposed by \eqref{mip:eq:modelinputconstraintsusum}, we do not need to formulate the constraints \eqref{mip:eq:setuptimeconstraintkmintau}, but we can instead use the simplified form \eqref{mip:eq:setuptimeconstraintkmintausumuniform} discussed in \remref{mip:rem:setuptimeconstraintkmintausumuniform}. Moreover, by imposing \eqref{mip:eq:modelinputconstraintsu} we can also use the simplified form \eqref{mip:eq:remsimplificationuminmax} instead of \eqref{mip:eq:setuptimeconstraintk}, see \remref{mip:rem:simplificationuminmax}. Thus, introducing the Boolean one-hot activators $\delta_k\in\Bs$ as in \eqref{mip:eq:delta} for $N_q=4$, we find the linear inequality constraints on the inputs $u_k$ related to the setup times to be given by
\begin{subequations}\label{mip:eq:modelsetuptimeconstraints}
\begin{align}
  J_{u}u_k &\leq \overline{u_\Sigma}\delta_{k}, \label{mip:eq:modelsetuptimeconstraintsk} \\
  J_{u}u_k &\leq \overline{u_\Sigma}S_\tau^\top\delta_{k-\tau}\quad \text{for all}\ \tau\in\Nset_{[1,\min\{k,3\}]}.
\end{align}
\end{subequations}

\subsection{MI-MPC for large-volume MR-HIFU hyperthermia}\label{mip:subsec:mpc}

\subsubsection{Prediction model}

Based on \eqref{mip:eq:model}, we define the prediction model as
\begin{equation}\label{mip:eq:mpcmodel}
  x_{i+1|k} = Ax_{i|k} + Bu_{i|k},
\end{equation}
where $x_{i|k}\in\Xs$ and $u_{i|k}\in\Us$ denote the predicted states and inputs, respectively, at $i\in\Nset$ time steps ahead of the prediction sequence's starting time $k\in\Nset$.

\subsubsection{Observer model}

To reduce the propagation of the measurement noise into the desired input $u_k$ computed by MPC, we use a Luenberger observer given by
  \begin{equation}\label{mip:eq:observer}
    \hat{x}_{k} = \left\{\begin{array}{ll} A\hat{x}_{k-1} + Bu_{k-1} + L(y_{k}-\hat{y}^-_{k}), & \text{if}\ \sigma_{k-1}\in\Qs, \\ A\hat{x}_{k-1} + Bu_{k-1}, & \text{if}\ \sigma_{k-1}\in\Es, \end{array}\right.
  \end{equation}
where $\hat{y}^-_{k} = A\hat{x}_{k-1} + Bu_{k-1}$ denotes the estimated output at time~$k$ before the measurement-based correction step in \eqref{mip:eq:observer} if possible by $\sigma_{k-1}\in\Qs$. The resulting estimation error dynamics are then given by
\begin{equation}\label{mip:eq:estimationerror}
  e_k = x_k-\hat{x}_k = \left\{\begin{array}{ll} (A-LA)e_{k-1} - Lv_k, & \text{if}\ \sigma_{k-1}\in\Qs, \\ Ae_{k-1}, & \text{if}\ \sigma_{k-1}\in\Es. \end{array}\right.
\end{equation}
The observer gain matrix is chosen as $L = 0.25 I_{n_x}$, which has been tuned such that the resulting estimation error dynamics \eqref{mip:eq:estimationerror} are stable, as can be shown using \thmref{mip:thm:observerstability} below, and exhibit desirable convergence behavior.

\begin{thm}\label{mip:thm:observerstability}
  The estimator error dynamics \eqref{mip:eq:estimationerror}, resulting from the system \eqref{mip:eq:model} and observer \eqref{mip:eq:observer}, are exponentially stable for $L=\alpha_lI$ with $0\leq\alpha_l\leq2$ and $v_k=0$, $k\in\Nset$.
\end{thm}

\begin{proof}
   Recall that $A$ in \eqref{mip:eq:model} is Schur due to the tissue's stable (first-order) thermal dynamics. In the absence of noise,
  Since $A$ is Schur, there exists a positive definite matrix $P>0$ such that
  \begin{equation}
    A^\top PA - P < 0,
  \end{equation}
  which ensures global exponential stability of the $\sigma_{k-1}\in\Es$ subsystem of \eqref{mip:eq:estimationerror} when using $V(e_k)=e_k^\top Pe_k$ as Lyapunov function. Next, for $L=\alpha_l I$ with $0\leq\alpha_l\leq2$, it holds that
  \begin{equation}
    (A-LA)^\top P(A-LA) - P = (1-\alpha_l)^2A^\top PA - P \leq A^\top PA - P < 0,
  \end{equation}
  which guarantees global exponential stability of the $\sigma_{k-1}\in\Qs$ subsystem of \eqref{mip:eq:estimationerror} using the same Lyapunov function $V(e_k)=e_k^\top Pe_k$. As $V(e_k)$ is a \emph{common} Lyapunov function for the two subsystems, \eqref{mip:eq:estimationerror} is globally exponentially stable for arbitrary switching sequences of $\sigma_k$. Finally, note that this Lyapunov function can also be used to establish mean-square stability in case of stochastic (Gaussian) noise.
\end{proof}

\subsubsection{MI-MPC optimization problem}

The temperature objectives of our MI-MPC setup are depicted schematically in \figref{mip:fig:Tmap} in cross-section perspective.
\begin{figure}[t]
  \centering
  \psfrag{tagmax}[cr][cr][1][0]{$\overline{T}+\epsilon$}
  \psfrag{tagmin}[cr][cr][1][0]{$\underline{T}-\underline{\epsilon}$}
  \includegraphics[width=0.7\textwidth,clip]{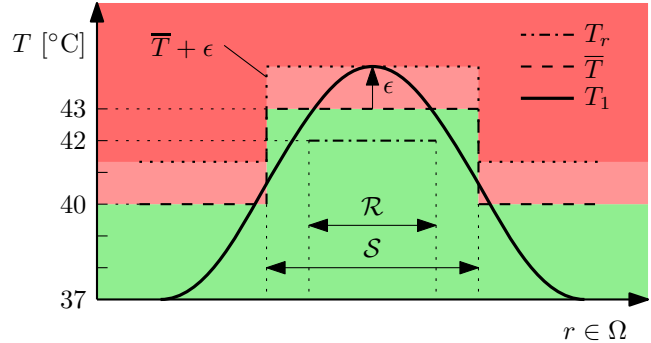}
  \caption{Schematic cross section of the temperature objectives corresponding to $\Rs$ and $\Ss$. The maximum violation $\epsilon$ is shown for some overheated temperature distribution $T_1$ such that $T_1\leq\overline{T}+\epsilon$.}
  \label{mip:fig:Tmap}
\end{figure}%
On $\Omega\subset\Rset^2$, denoting the patient domain in the focal plane, $T_r:\Rs\to\Rset$ (dash-dotted) is the reference temperature of 42 $^\circ$C in the ROI $\Rs$, and $\overline{T}:\Omega\to\Rset$ is the location-dependent upper temperature bound (dashed) used to prevent overheating (red). To translate these objectives to the state-space representation \eqref{mip:eq:model}, we introduce the performance variables $z_k=Hx_k\in\Rset^{n_z}$, with $H\in\{0,1\}^{n_z\times n_x}$ being a matrix with one $1$ per row (and at most one 1 per column), which are the temperatures of the $n_z<n_x$ voxels inside $\Rs$. Furthermore, we use $z_r\in\Rset^{n_z}$ and $\overline{x}\in\Rset^{n_x}$ to denote the voxel-wise temperature reference and upper bounds corresponding to the values of $T_r$ and $\overline{T}$, respectively. The maximum violation of the upper temperature bound is measured by the slack variable $\epsilon_k = \epsilon(x_k) = \lVert\max\{x_k-\overline{x},0_{n_x}\}\rVert_\infty\in\Rset_{\geq0}$, where the maximum operator is used element-wise. The predicted performance and slack variables are denoted by $z_{i|k}$ and $\epsilon_{i|k}$, respectively.

The resulting MI-MPC optimization problem, which can be easily written in the form of \eqref{mip:eq:OCP}, is then given by
\begin{subequations}\label{mip:eq:opt}
\begin{equation}\label{mip:eq:optcostQP}
  \min_{\bm{\delta}_k,\bm{u}_k} \sum_{i=0}^{N} (z_{i|k}-z_r)^\top Q (z_{i|k}-z_r) + f_{\epsilon}\epsilon_{i|k},
\end{equation}
subject to
\begin{align}
  x_{i+1|k} &= Ax_{i|k} + Bu_{i|k}, && \forall\ i\in\Nset_{[0,N-1]}, \label{mip:eq:optconstraintsdynamics} \\
  x_{0|k} &= \hat{x}_{k}, && \label{mip:eq:optconstraintsx0}  \\
  1_{N_q}^\top \delta_{i|k} &= 1, && \forall\ i\in\Nset_{[0,N-1]}, \label{mip:eq:optconstraintsdeltasum}  \\
  x_{i|k} & \leq \overline{x} + 1_{n_x}\epsilon_{i|k}, && \forall\ i\in\Nset_{[0,N]}, \label{mip:eq:optconstraintsTminmax} \\
  0 &\leq \epsilon_{i|k}, && \forall\ i\in\Nset_{[0,N]}, \label{mip:eq:optconstraintsepsmin} \\
  0_{n_u} & \leq u_{i|k} \leq \overline{u}, && \forall\ i\in\Nset_{[0,N-1]},  \label{mip:eq:optconstraintsinput} \\
  J_{u}u_{i|k} &\leq \overline{u_\Sigma}\delta_{i|k}, && \forall\ i\in\Nset_{[0,N-1]}, \label{mip:eq:optconstraintssetuptime0} \\
  J_{u}u_{i|k} &\leq \overline{u_\Sigma}S^\top_\tau\delta_{i-\tau|k}, &&\hspace{-4mm} \left\{\begin{aligned} &\forall\hspace{1mm}\tau\hspace{-0.5mm}\in\Nset_{[1,\min\{k,\overline{s}\}]}, \\ &\forall\ i\in\Nset_{[0,N-1]}. \end{aligned}\right. \label{mip:eq:optconstraintssetuptimetau}
\end{align}
\end{subequations}
In \eqref{mip:eq:optcostQP}, we choose $Q = \frac{1}{n_z} I_{n_z}$ and $f_{\epsilon} = 10$, which are normalized with respect to the number of weighted variables (note that $\epsilon$ is scalar) for more intuitive tuning, and use horizon $N=8$. The weighting $f_{\epsilon}$ on the linear term in \eqref{mip:eq:optcostQP}, incorporating the upper temperature bound as a soft constraint, is large relative to $Q$, reflecting the fact that the prevention of overheating is prioritized. The quadratic term in \eqref{mip:eq:optcostQP} weighted by $Q$ enforces tracking of the ROI temperature towards the optimal treatment temperature. The equality constraints \eqref{mip:eq:optconstraintsdynamics} capture the system dynamics as modeled by \eqref{mip:eq:mpcmodel}, with the initial condition \eqref{mip:eq:optconstraintsx0} determined by the observer \eqref{mip:eq:observer}. The equality \eqref{mip:eq:optconstraintsdeltasum} captures \eqref{mip:eq:deltasum} throughout the horizon. The inequalities \eqref{mip:eq:optconstraintsTminmax} encode the upper temperature bound as a soft constraint using the slack variable, which is restricted to be nonnegative by \eqref{mip:eq:optconstraintsepsmin}. Finally, \eqref{mip:eq:optconstraintsinput} corresponds to the actuator constraints \eqref{mip:eq:modelinputconstraintsu}, and the inequalities \eqref{mip:eq:optconstraintssetuptime0}-\eqref{mip:eq:optconstraintssetuptimetau} incorporate the mode and setup time constraints \eqref{mip:eq:modelsetuptimeconstraints}. Note that in \eqref{mip:eq:optconstraintssetuptimetau}, $\delta_{i-\tau|k}=\delta_{k+i-\tau}$ for $i<\tau$, following from the system's activator sequence prior to time $k\in\Nset$. As discussed in \algoref{mip:alg:MIMPC}, after solving \eqref{mip:eq:opt} we apply the admissibility assurance from \eqref{mip:eq:thmadmissassurance} in \thmref{mip:thm:admissassurance} to ensure that the resulting optimal pair $(\bm{\delta}_k^*,\bm{u}_k^*)$ is $\Sigma$-admissible.

\subsection{Simulation results}

The temperature control performance of the MI-MPC setup is verified by means of simulation, using Matlab R2017b and Gurobi 8.1.1. We initialize the transducer in cell 1, set the plant, observer, and controller states to zero, corresponding to the monitored tissue to be at the blood temperature $T_b=37$~$^\circ$C before treatment, and initialize the noise sequence. The temperature evolution of the ROI voxels of the plant (gray) and the observer/controller (black) are visualized in \figref{mip:fig:corner4lumpin_TROI} using their mean (solid) and maximum/minimum values (dashed).
\begin{figure}[t]
  \centering
  \psfrag{z+Tb}[cc][cc][0.9][0]{$z+T_b$}
  \psfrag{s+Tb}[cc][cc][0.9][0]{$\hat{z}+T_b$}
  \psfrag{time}[cc][cc][0.9][0]{Time [s]}
  \psfrag{temperature}[cc][cc][0.9][0]{Temperature in ROI [$^\circ$C]}
  \psfrag{0}[cc][cc][0.9][0]{0}
  \psfrag{100}[cc][cc][0.9][0]{100}
  \psfrag{200}[cc][cc][0.9][0]{200}
  \psfrag{300}[cc][cc][0.9][0]{300}
  \psfrag{400}[cc][cc][0.9][0]{400}
  \psfrag{37}[cc][cc][0.9][0]{37}
  \psfrag{38}[cc][cc][0.9][0]{38}
  \psfrag{39}[cc][cc][0.9][0]{39}
  \psfrag{40}[cc][cc][0.9][0]{40}
  \psfrag{41}[cc][cc][0.9][0]{41}
  \psfrag{42}[cc][cc][0.9][0]{42}
  \psfrag{43}[cc][cc][0.9][0]{43}
  \includegraphics[scale=0.75,clip]{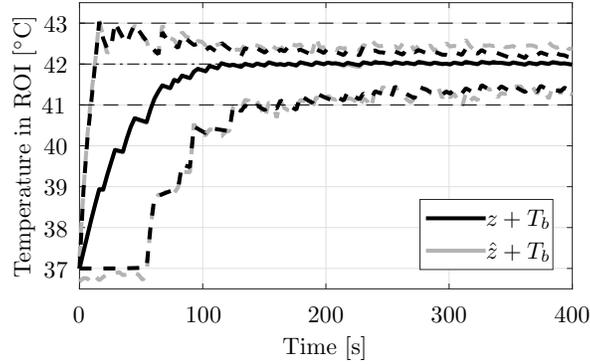}
  \caption{The mean (solid) and maximum/minimum (dashed) temperature inside the ROI $\Rs$ of the plant (black) and as estimated by the observer (gray).}
  \label{mip:fig:corner4lumpin_TROI}
\end{figure}%
It can be seen that after approximately 150~s the entire ROI is heated to 41~$^\circ$C, which is the temperature above which the beneficial hyperthermia-related effects are adequately triggered. Furthermore, during heat-up the upper temperature limit is observed to be reached, but no overheating occurs, and the average ROI temperature converges to the optimum $42$~$^\circ$C.

In \figref{mip:fig:corner4lumpin_Ptot}, the total acoustic input power per cell is shown, together with the destination mode $\post(\sigma_k)$.
\begin{figure}[t]
  \centering
  \psfrag{x}[cc][cc][0.9][0]{Time [s]}
  \psfrag{y}[cc][cc][0.9][0]{Acoustic power [W]}
  \psfrag{sigma}[cc][cc][0.9][0]{Destination mode $\post(\sigma_k)$}
  \psfrag{11}[cc][cc][0.9][0]{$1$}
  \psfrag{22}[cc][cc][0.9][0]{$2$}
  \psfrag{33}[cc][cc][0.9][0]{$3$}
  \psfrag{44}[cc][cc][0.9][0]{$4$}
  \psfrag{postsss}[cc][cc][0.9][0]{$\post(\sigma_k)$}
  \psfrag{0}[cc][cc][0.9][0]{0}
  \psfrag{100}[cc][cc][0.9][0]{100}
  \psfrag{200}[cc][cc][0.9][0]{200}
  \psfrag{300}[cc][cc][0.9][0]{300}
  \psfrag{400}[cc][cc][0.9][0]{400}
  \psfrag{20}[cc][cc][0.9][0]{20}
  \psfrag{40}[cc][cc][0.9][0]{40}
  \psfrag{60}[cc][cc][0.9][0]{60}
  \psfrag{80}[cc][cc][0.9][0]{80}
  \psfrag{50}[cc][cc][0.9][0]{50}
  \psfrag{100}[cc][cc][0.9][0]{100}
  \psfrag{1}[cc][cc][0.9][0]{1}
  \psfrag{2}[cc][cc][0.9][0]{2}
  \psfrag{3}[cc][cc][0.9][0]{3}
  \psfrag{4}[cc][cc][0.9][0]{4}
  \includegraphics[scale=0.58,clip]{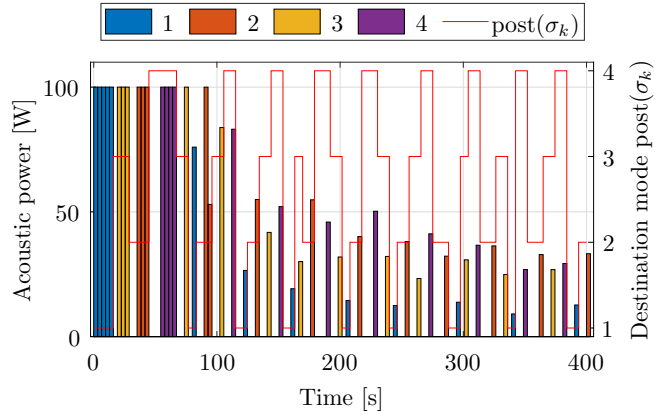}
  \caption{The total acoustic power per cell, and the destination mode $\post(\sigma_k)$ (red line).}
  \label{mip:fig:corner4lumpin_Ptot}
\end{figure}%
First, note that $\post(\sigma_k)$ (which is directly related to $\delta_k$ by \eqref{mip:eq:delta}) and $u_k^q$ indeed correspond to an admissible sequence of actuator-input pairs, describing the SAcSS to heat only in the active cell, perform the minimum number of mode switches required to allow for the desired nonzero inputs, and respect the setup time induced by each switch. From this, we verify the admissibility-assured MI-MPC setup to be functioning as intended. Next, regarding the input power, the figure shows that the controller initially requests maximum power for fast heat-up, reaching the upper limit of the total cell power constraint \eqref{mip:eq:optconstraintssetuptime0}. Moreover, in this period it often heats a certain cell for several consecutive samples, thereby reducing actuator downtime, thus contributing to achieving a short heat-up phase. For $t_k>100$~s, however, each cell is only heated for one sample before continuing to the next cell, as this allows for maintaining a ROI temperature distribution that is as homogeneous as possible.

The plant temperature at $t_k=400$~s is shown in \figref{mip:fig:corner4lumpin_Tmap400}, where additionally the perimeters of $\Rs$ (red) and $\Ss$ (dashed black) are plotted.
\begin{figure}[t!]
  \centering
  \psfrag{x}[cc][cc][0.9][0]{$r_x$ [m]}
  \psfrag{y}[cc][cc][0.9][0]{$r_y$ [m]}
  \psfrag{-0.02}[cc][cc][0.9][0]{$-0.02$}
  \psfrag{0.01}[cc][cc][0.9][0]{$-0.01$}
  \psfrag{0}[cc][cc][0.9][0]{0}
  \psfrag{0.01}[cc][cc][0.9][0]{$0.01$}
  \psfrag{0.02}[cc][cc][0.9][0]{$0.02$}
  \psfrag{1}[cc][cc][0.7][0]{38}
  \psfrag{2}[cc][cc][0.7][0]{39}
  \psfrag{3}[cc][cc][0.7][0]{40}
  \psfrag{4}[cc][cc][0.7][0]{41}
  \psfrag{4.5}[cc][cc][0.7][0]{41.5}
  \psfrag{5}[cc][cc][0.7][0]{42}
  \psfrag{5.1}[cc][cc][0.7][0]{42.1}
  \psfrag{5.2}[cc][cc][0.7][0]{42.2}
  \includegraphics[scale=0.80,clip]{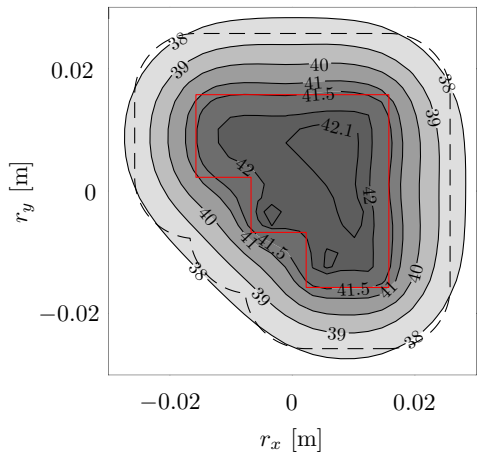}
  \caption{Contour plot of the focal plane temperature distribution in the ROI $\Rs$ (red) and $\Ss$ (dashed) at $t_k=400$~s.}
  \label{mip:fig:corner4lumpin_Tmap400}
\end{figure}%
This figure clearly illustrates the temperature homogeneity in $\Rs$. The maximum temperature (42.1~$^\circ$C) is observed in all cells except cell 2 (top-left). In \figref{mip:fig:corner4lumpin_Ptot}, it can be seen that at time $t_k=400$~s the transducer has just arrived at cell 2, and that the MPC has computed significant sonication power (approximately 33~W) to be applied, which will result in a temperature increase in this cell. This exemplifies the advantage of a \emph{predictive} controller for large-volume MR-HIFU hyperthermia. That is, the MPC is able to anticipate the diffusive heat loss during the future time interval in which a cell cannot be heated due to the transducer having moved to another cell. To preemptively counteract this heat loss, which is most severe at the ROI's corners, additional heat is injected especially at the corners just before relocating the transducer, resulting in temperature peaks (slightly) above the reference temperature of 42 $^\circ$C. Then, after the subsequent heat diffusion during the time period that a cell cannot be heated, minor temperature peaks remain near the ROI's corners, while the corner temperatures are still within the desired range ($\geq41$~$^\circ$C), as can be seen in the corners of cells 1, 3, and 4 in \figref{mip:fig:corner4lumpin_Tmap400}.

\subsection{Computational efficiency}

To evaluate the improved computational efficiency of the online MI-MPC problem when using the modeling approach proposed in this paper, we compare the resulting computation time to that of the MI-MPC setup obtained by modeling the SAcSS using a lifting approach \cite{Subramanian2012}, effectively yielding a constrained switched linear system (cSLS) \cite{Liberzon2003,Philippe2016} in MLD form \cite{Bemporad1999}. Similar to our method, for the purpose of MI-MPC this also requires duplicating shared input channels such that the resulting input vector is the collection of the disjoint inputs, and hence the discrete-time dynamics are described by \eqref{mip:eq:model}. Contrary to our method, however, is the fact that we must now define $s_{q\tilde{q}}$ actuator \emph{transition} modes for each switch $(q,\tilde{q})\in\Es$, in addition to the $N_q=4$ operational modes, which for this case study results in $s_\Sigma \coloneqq \sum_{(q,\tilde{q})\in\Es} s_{q\tilde{q}} = 24$ additional modes.

Correspondingly, we define the $N_q$ Boolean actuator-operation states
\begin{equation}\label{mip:eq:cSLSactuatoractivity}
  \beta_k^q \in \{0,1\},\quad q\in\Qs,
\end{equation}
for which $\beta_k^q=1$ if at time $k$ the SAcSS is in operational mode $q\in\Qs$, and $\beta_k^q=0$ otherwise. In addition, we must for each actuator switch $(q,\tilde{q})\in\Es$ with nonzero setup time $s_{q\tilde{q}}$ define the Boolean actuator-transition states
\begin{equation}\label{mip:eq:cSLSactuatortransition}
  \beta_k^{q\tilde{q}} = \begin{bmatrix} \beta_k^{q\tilde{q},1} \\ \vdots \\ \beta_k^{q\tilde{q},s_{q\tilde{q}}} \end{bmatrix} \in\{0,1\}^{s_{q\tilde{q}}},
\end{equation}
where $\delta_k^{q\tilde{q},i}=1$ if at time $k$ the SAcSS is at setup time instant $i\in\Nset_{[1,s_{q\tilde{q}}]}$ of the switch $(q,\tilde{q})\in\Es$. Since a SAcSS' actuator can only be in one mode or at one stage of a switch at any given time, it must hold that
\begin{equation}\label{mip:eq:cSLSbetasum}
  1_{N_q+s_\Sigma}^\top \beta_k = 1,
\end{equation}
where $\beta_k$ is the collection of all actuator-operation and actuator-transition states in \eqref{mip:eq:cSLSactuatoractivity} and \eqref{mip:eq:cSLSactuatortransition} at time $k$. The mode-dependent input constraints can be expressed as
\begin{equation}\label{mip:eq:cSLSinputconstraint}
  \underline{u}^q\beta_k^q \leq u_k^q \leq \overline{u}^q\beta_k^q \quad \text{for all}\ q\in\Qs,
\end{equation}
or, in case the input bounds are explicitly imposed by \eqref{mip:eq:modelinputconstraintsu} and using $\overline{u_\Sigma}$ from \eqref{mip:eq:modelinputconstraintsusum}, as the simplified form
\begin{equation}\label{mip:eq:cSLSinputconstraintsum}
  1_{n_u^q}^\top u_k^q \leq \overline{u_\Sigma}\beta_k^q \quad \text{for all}\ q\in\Qs.
\end{equation}

Next, the switching progression along the successive actuator-transition states is prescribed by
\begin{equation}\label{mip:eq:cSLSactuatortransitiondynamics}
  \begin{bmatrix} \beta_{k+1}^{q\tilde{q},2} \\ \vdots \\ \beta_{k+1}^{q\tilde{q},s_{q\tilde{q}}} \end{bmatrix} = \begin{bmatrix} 1 & 0 & & \\ & \ddots & \ddots & \\ & & 1 & 0 \end{bmatrix} \beta_k^{q\tilde{q}},
\end{equation}
for all $(q,\tilde{q})\in\Es$ with $s_{q\tilde{q}}>1$. Furthermore, recall that a SAcSS can only be in operational actuator mode $q\in\Qs$ when at the previous time instant it was either already in this mode, in its last setup time instant towards this mode, or in an operational mode $\tilde{q}\in\Qs$ from which the switch towards $q$ induces zero setup time. This can be described by the inequality
\begin{equation}\label{mip:eq:cSLSmoveconstraintarrive}
  \beta_k^q \leq \beta_{k-1}^q + \sum_{\tilde{q}\in\Qs_0^q} \beta_{k-1}^{\tilde{q}} + \sum_{\tilde{q}\in\Qs_{>0}^q} \beta_{k-1}^{\tilde{q}q,s_{\tilde{q}q}},
\end{equation}
where similarly to \eqref{mip:eq:setupmodeset} we define for $q\in\Qs$ the sets
\begin{align*}
  \Qs_{0}^{q} &= \{ \tilde{q}\in\Qs \mid s_{\tilde{q}q} = 0 \}, \\ 
  \Qs_{>0}^{q} &= \{ \tilde{q}\in\Qs \mid s_{\tilde{q}q} > 0 \}, 
\end{align*}
as the set of actuator modes $\tilde{q}\in\Qs$ from which the switch towards mode $q\in\Qs$ induces zero and nonzero setup time, respectively. Similarly, an actuator switch $(q,\tilde{q})\in\Es$ can only start at time $k$ when at $k-1$ the SAcSS was in or arriving at mode $q\in\Qs$, and hence
\begin{equation}\label{mip:eq:cSLSmoveconstraintstart}
  \beta_k^{q\tilde{q},1} \leq \beta_{k-1}^q + \sum_{\hat{q}\in\Qs_{>0}^q} \beta_{k-1}^{\hat{q}q,s_{\hat{q}q}}.
\end{equation}
Note that, contrary to \eqref{mip:eq:cSLSmoveconstraintarrive}, in \eqref{mip:eq:cSLSmoveconstraintstart} we do not need to include the possibility that at time $k-1$ the SAcSS was in some mode $\hat{q}\in\Qs_0^q$, since switching towards mode $\tilde{q}$ at time $k$ would then be described by the direct switch $(\hat{q},\tilde{q})\in\Es$ (i.e., $\beta_k^{\hat{q}\tilde{q}}=1$, not $\beta_k^{q\tilde{q}}=1$) due to $\Gamma$ being complete and satisfying the triangle inequality.

Finally, to obtain the MLD system form of the cSLS model of the considered SAcSS, all of the above constraints must be combined and rewritten in matrix-vector form in terms of the inputs and the actuator-operation and actuator-transition states, and combined with the discrete-time plant dynamics \eqref{mip:eq:model}. However, as these are not easy to write down compactly, the MLD system description is not explicitly given here.

The resulting MI-MPC consists of the optimization problem \eqref{mip:eq:optcostQP} (except for optimizing over $\bm{\beta}_k=(\beta_{0|k},\ldots,\beta_{N-1|k})$ instead of $\bm{\delta}_k$) subject to the thermal dynamics and initial condition constraints \eqref{mip:eq:optconstraintsdynamics} and \eqref{mip:eq:optconstraintsx0}, the upper temperature bound soft constraints \eqref{mip:eq:optconstraintsTminmax} and \eqref{mip:eq:optconstraintsepsmin}, the input bounds \eqref{mip:eq:optconstraintsinput}, and additionally (after substituting $u_k$ and $\beta_{k}$ by their prediction counterparts $u_{i|k}$ and $\beta_{i|k}$) for all $i\in\Nset_{[0,N-1]}$ the one-hot encoding constraint \eqref{mip:eq:cSLSbetasum}, the mode-dependent input constraints \eqref{mip:eq:cSLSinputconstraintsum}, the switch progression constraints \eqref{mip:eq:cSLSactuatortransitiondynamics}, and the constraints \eqref{mip:eq:cSLSmoveconstraintarrive} and \eqref{mip:eq:cSLSmoveconstraintstart} describing the actuator possibilities at the end and beginning of an actuator switch. Note that (for all $i\in\Nset_{[0,N-1]}$) we must impose \eqref{mip:eq:cSLSinputconstraintsum} and \eqref{mip:eq:cSLSmoveconstraintarrive} for all $q\in\Qs$, and \eqref{mip:eq:cSLSactuatortransitiondynamics} and \eqref{mip:eq:cSLSmoveconstraintstart} for all $(q,\tilde{q})\in\Es$ with $s_{q\tilde{q}}>1$ and $s_{q\tilde{q}}>0$, respectively.

Compared to \eqref{mip:eq:opt}, the continuous part of the optimization problem is unchanged, while the number of Boolean decision variables increases from $N_qN=32$ ($\delta_{i|k}$ in \eqref{mip:eq:deltavec} over the horizon $N$) to $(N_q+s_\Sigma)N=224$ ($\beta_{i|k}$ in \eqref{mip:eq:cSLSactuatoractivity}-\eqref{mip:eq:cSLSactuatortransition} over $N$). The number of integer equality constraints increases from $N=8$ in \eqref{mip:eq:optconstraintsdeltasum} to $(1+\sum_{(q,\tilde{q})\in\Es}\max\{s_{q\tilde{q}}-1,0\})N=104$ in \eqref{mip:eq:cSLSbetasum} and \eqref{mip:eq:cSLSactuatortransitiondynamics} over $N$. Finally, the number of mixed-integer inequality constraints increases from $N_q(1+\overline{s})N=128$ in \eqref{mip:eq:optconstraintssetuptime0}-\eqref{mip:eq:optconstraintssetuptimetau} (actually $N_q(1+\overline{s}-2/4)N=112$, since 2 out of 4 constraints generated using $S_3$ from \eqref{mip:eq:modelStau} for each $i$ in \eqref{mip:eq:optconstraintssetuptimetau} are redundant, see \remref{mip:rem:redundantconstraints}) to $(2N_q+\sum_{(q,\tilde{q})\in\Es}\min\{s_{q\tilde{q}},1\})N=160$ in \eqref{mip:eq:cSLSinputconstraintsum} and \eqref{mip:eq:cSLSmoveconstraintarrive}-\eqref{mip:eq:cSLSmoveconstraintstart} over $N$.

In simulation, the control inputs and resulting temperatures using the cSLS-based MI-MPC are found to indeed be exactly the equal to those obtained using the MI-MPC \eqref{mip:eq:opt} up to numerical tolerances.
The computation times for both MI-MPC setups have been recorded using five different noise realizations and five runs per realization, i.e., twenty-five simulations for each MI-MPC. The distributions of the computation times are depicted in \figref{mip:fig:corner4lumpin_Tcomp} (normalized with respect to the largest observed computation time) over the entire simulation, and over the transient ($t_k<100$) and steady-state ($t_k\geq100$) intervals separately.
\begin{figure}[t]
  \centering
  \psfrag{Norm. computation time}[cc][cc][0.9][0]{Norm. computation time}
  \psfrag{SAcSS}[cc][cc][0.8][0]{SAcSS}
  \psfrag{cSLS}[cc][cc][0.8][0]{cSLS}
  \psfrag{SAcSS tr}[cc][cc][0.8][0]{SAcSS,t}
  \psfrag{cSLS tr}[cc][cc][0.8][0]{cSLS,t}
  \psfrag{SAcSS ss}[cc][cc][0.8][0]{SAcSS,s}
  \psfrag{cSLS ss}[cc][cc][0.8][0]{cSLS,s}
  \psfrag{0}[cc][cc][0.8][0]{0}
  \psfrag{0.2}[cc][cc][0.8][0]{0.2}
  \psfrag{0.4}[cc][cc][0.8][0]{0.4}
  \psfrag{0.6}[cc][cc][0.8][0]{0.6}
  \psfrag{0.8}[cc][cc][0.8][0]{0.8}
  \psfrag{1}[cc][cc][0.8][0]{1}
  \includegraphics[scale=0.80,clip]{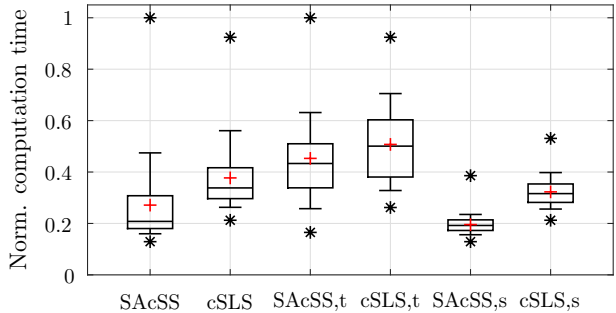}
  \caption{Distribution of the normalized computation times for the SAcSS MI-MPC \eqref{mip:eq:opt} and cSLS MI-MPC over the entire simulation, and over the transient (t) and steady-state (s) intervals only. We indicate the mean (red +), median (central mark), 25th and 75th percentiles (box edges), 10th and 90th percentiles (whiskers), and extrema ($\ast$).}
  \label{mip:fig:corner4lumpin_Tcomp}
\end{figure}
The MI-MPC \eqref{mip:eq:opt} designed specifically for SAcSS clearly outperforms the MI-MPC derived using the cSLS model, with the latter showing 39\% and 63\% larger overall mean and median computation times. 
In fact, in the constant-temperature phase, which comprises the majority of a hyperthermia treatment, the cSLS MI-MPC computation times' 10th percentile is larger than the 90th percentile of \eqref{mip:eq:opt}, and the mean and median computation times are 65\% larger than those obtained using our novel method.

\section{Conclusion}\label{mip:sec:conclusion}

In this paper, the class of switched-actuator systems with setup times (SAcSSs) has been formally introduced. As key contribution, a modeling framework for SAcSSs has been presented, which is specifically tailored to (a) allow for user-friendly system specification, modeling, and controller synthesis, and (b) yield compact models leading to efficient MI-MPCs. In particular, the resulting model is in MLD form, consisting of a state-space representation of the SAcSS' dynamics in its $N_q$ operational modes, combined with systematically derived mixed-integer linear inequality constraints on the inputs (thus achieving compatibility with MI-MPC) to incorporate the mode switching and setup times. A distinctive property of our method is that instead of explicitly modeling the zero-input ``switching modes'' as in a constrained switched linear system (cSLS) or lifting approach, which would require many auxiliary Boolean variables per time step, it uses only the $N_q$ Boolean variables corresponding to the operational modes, and infers the actuator activity from a sequence of these variables in the prediction horizon. This reduces the MI-MPC's computational complexity. The proposed modeling procedure and corresponding MI-MPC setup have been validated in a large-volume MR-HIFU hyperthermia case study. It is demonstrated that the desired temperature distribution can be achieved and maintained in a large tumor by coordinated heating and mechanical transducer displacement, and that the MI-MPC's computational efficiency is improved with respect to a cSLS/lifting approach.


\section*{Acknowledgements}
The authors thank K.S. Mohan and J. van Wordragen for their contributions to the case study, and J.B. Rawlings for his valuable insights and suggestions.

\section*{References}

\bibliography{C:/Users/ddeenen/SURFdrive/PhD/Literature/BibTeX/library}

\end{document}